\documentclass[ leqno]{amsproc}
\usepackage{cancel}
\usepackage{textcomp}
\usepackage{authblk}
\usepackage{amssymb}
\usepackage{epstopdf}
\usepackage[sans]{dsfont}
\usepackage[applemac]{inputenc}
\usepackage[english]{babel}
\usepackage{latexsym}
\usepackage{subfigure}
\usepackage{color}
\usepackage{float}
\usepackage{amsmath}
\usepackage{amsfonts}
\usepackage{graphicx}

\newtheorem{theorem}{Theorem}
\numberwithin{equation}{section}
\newtheorem{lemma}[theorem]{Lemma}

 \newtheorem{proposition}[theorem]{Proposition}
\newtheorem{corollary}[theorem]{Corollary}
\newtheorem{remark}[theorem]{Remark}
\newtheorem{example}[theorem]{Example}
\newtheorem{definition}[theorem]{Definition}

\newtheorem{assumption}[theorem]{Assumption}

\def\ker{\textrm{Ker}}

\def\RR{\mathbb R}

\def\be{\begin{equation}}
\def\ee{\end{equation}}
\def\bea{\begin{eqnarray}}
\def\eea{\end{eqnarray}}

\begin{document}

\begin{center} {\Large\bf Microscopic modeling and analysis of collective decision making: equality bias leads suboptimal solutions}
\bigskip

{Pierluigi Vellucci$^1$ \and Mattia Zanella$^2$}

\end{center}

\bigskip
{\it Abstract:}
We discuss a novel microscopic model for collective decision-making interacting multi-agent systems. In particular we are interested in modeling a well known phenomena in the experimental literature called equality bias, where agents tend to behave in the same way as if they were as good, or as bad, as their partner. We analyze the introduced problem and we prove the suboptimality of the collective decision-making in the presence of equality bias. Numerical experiments are addressed in the last section.

\section{Introduction}


Several{\let\thefootnote\relax\footnotetext{$^1$Department of Economics, Roma Tre University, via Silvio D'Amico 77, 00145 Rome, Italy; pierluigi.vellucci@uniroma3.it}} experimental{\let\thefootnote\relax\footnotetext{$^2$Politecnico di Torino, Department of Mathematical Sciences, Corso Duca degli Abruzzi 24, 10129 Torino, Italy; mattia.zanella@polito.it}} works on group psychology has been done in recent years in order to observe unexpected dysfunctional behaviors in decision-making communities, see \cite{BOLRRF,G,KD,MPNAS} and the references therein. Usual example are the \emph{groupthink}, a collective phenomena whereby people try to minimize internal conflicts for reaching consensus to the detriment of the common good, the \emph{Dunning-Kruger effect}, regarding an overestimation of personal competence of unskilled people, and the \emph{equality bias}, whereby people behave as if they are as good, or as bad as their partner. In the following we will focus on this latter aspect of decision--making systems.

A valuable improvement on the direction of understanding the emergence of the equality bias has been done in \cite{MPNAS}. Here, authors asked how people deal with individual differences in competence in the context of a collective perceptual decision-making task, developing a metric for estimating how participants weight their partner's opinion relative to their own. Empirical experiments, replicated across three slightly different countries like Denmark, Iran, and China, show how participants assigned nearly equal weights to each other's opinions regardless of the real differences in their competence. The results show that the equality bias is particularly costly for a group when a competence gap separates its members.

Drawing inspiration by these recent experimental results, and by the mathematical set-up introduced in the recent works \cite{AHP,APZa,APZc,BT,CS,PT1,PT2,T}, we consider here a  microscopic model taking into account the influence of the competence in collective decision-making tasks for systems of interacting agents. This works follows the recent study of the authors \cite{PVZ} where the decision-making task is discussed at the kinetic level.

The approach proposed in this paper is based on the Laplacian matrix of the connectivity graph and is inspired by classical works on self--organization \cite{CS,Shen}. With reference to the experimental literature we introduce competence--based interaction functions describing the maximum competence (MC) and the equality bias (EB) case. In particular, the MC model sketches the case in which the emerging decision coincides with the one of the most competent agent. On the other hand the EB model should deal with the complementary case. Based on a simplified communication coefficients, we derive the asymptotic convergence of the overall system for the decision models. A key feature of present modeling is the evolution of the competence variable, whose dynamics takes into account the social background of the single agent and the possibility to improve specific competences during interactions with more competent agents, see \cite{BT,PT1}.
At the continuous level it has been showed in \cite{PVZ} how the variation of the mean opinion of agents with given competence follows the choice of the most competent agents in the MC case. The present approach is based on the explicit derivation of eigenvalues of the system. 

The present manuscript is organized as follows. In Section \ref{sec:micro} we briefly review some microscopic models for alignment dynamics, we introduce here a specific model for decision and competence. Then we discuss two main models for the collective decision-making, the mentioned MC model and the EB model. In Section \ref{sec:simpl} we analyze the main properties of the model and we show how the equality bias leads the system of agents toward suboptimal collective decisions, computing the eigenvalues of the aforementioned Laplacian matrix and proving that the collective decision-making in the presence of equality bias is suboptimal for each $t>0$. Finally, in Section \ref{sec:num} we address numerical experiments based on the introduced model.

\section{Description of the model}\label{sec:micro}

In this section we discuss some modeling aspects of second order microscopic model for decision-making dynamics. Our mathematical approach follows the set-up of several recent works on opinion dynamics, see \cite{APZa,BT,CS,T} and the references therein. These class of models gained deepest attention in scientific research in the last decade thanks to their countless applications in biology, socio-economic sciences and control theory \cite{AHP,APTZ,APZb,APZc,CC,CFL,CFTV,CPT,DMPW,HZ,PT2}.

\subsection{Microscopic models for the collective behavior}
Without intending to review whole literature, we introduce some well-known microscopic models  describing particular aspects of the aggregate motion of a finite system of interacting agents. We focus in particular on alignment-type dynamics.

We are interested in studying the dynamics of $N\in\mathbb{N}$ individuals with the following general structure at time $t\in\RR^+$
\begin{equation}\begin{cases}\label{eq:initial_val}
\dot{x}_i = f(x_i,w_i), \qquad i=1,\dots,N,\\
\dot{w}_i = S(x_i)+\displaystyle\dfrac{1}{\alpha_i}\sum_{j=1}^NP(x_i,x_j;w_i,w_j) (w_j-w_i),
\end{cases}\end{equation}
where $(x_i,w_i)\in\RR^{2d}$ for each $t\ge 0$, $S(\cdot)$ is a self-propelling term and $P(\cdot,\cdot;\cdot,\cdot)$ is a general interaction function depending on both the considered variables. In \eqref{eq:initial_val} we introduced a function $f:\RR^{2d}\rightarrow \RR^{d}$, it assumes the form $f(x_i,w_i)=w_i$ in case of flocking systems, in this case $x_i,w_i$ are the space and velocity variables the $i$th agent. It may describe a wider class of processes which will be specified later on. \\

We exemplify the structure of flocking systems by presenting the Cucker-Smale (CS) model and the Motsch-Tadmor (MT) model. In the classic CS model each agent adjusts its velocity by adding a weighted average of the differences of its velocity with those of all the other agents. Therefore, for all $i\in\{1,\dots,N\}$ we consider a symmetric interaction function of the form
\begin{equation}\label{eq:P_dist}
P(x_i,x_j;w_i,w_j)=p(\|x_i-x_j \|^2)
\end{equation}
depending on the Euclidean distance between agents and the constant scaling $\alpha_i=N$, see \cite{CS}. In particular the typical choice is the following
\begin{equation*}
p(\|x_i-x_j \|^2) = \dfrac{K}{(\zeta^2+\|x_i-x_j \|^2)^{\gamma}},
\end{equation*}
with $K,\zeta>0$ and $\gamma\ge 0$. Without considering self-propelling terms, i.e. $S(\cdot)\equiv 0$, it has been shown how under these assumptions that the resulting initial value problems is well-posed: mass and momentum are preserved and the solution has compact support for both position and velocity \cite{CCR,CFRT,CFTV}. Further, in the CS model unconditional alignment emerges for $\gamma\le1/2$ and the velocity support collapses exponentially to a single point and the system holds the same disposition. \\

An example of non-symmetric interactions in flocking systems is given by the MT model \cite{MT1,MT2}. Here the alignment is based on the relative influence between the system of agents, therefore we consider an interaction of the form introduced in \eqref{eq:P_dist}
whereas the scaling factors $\alpha_i>0$ are given by
\begin{equation*}
\alpha_i = \sum_{j\ne i}P(\|x_i-x_j \|^2).
\end{equation*}
With this definition the dynamics looses any property of symmetry of the CS model, linking the initial value problem \eqref{eq:initial_val} to more sophisticated models where the $i$th agent may interact with the $j$th agent but not vice versa, for example leader-follower models as well as limited perception models \cite{DOCBC}.

\subsection{A competence-based model for collective decision-making}

We are interested to describe the coupled evolution of decisions and competence in a system of $N\in\mathbb{N}$ interacting agents. Each agent is endowed with two quantities $(x_i,w_i)$ representing its competence and decision respectively, where $x_i\in X\subseteq\RR^+$ and $w_i\in[-1,1]=\mathcal{I}$, where $\pm 1$ denote two opposite  possible decisions of an agent.

One of the main factors influencing the evolution of the competence variable is the social background in which individuals lives. It is therefore natural to assume that competence is partially inherited from the environment with the possibility to learn specific competences by interacting with more competent agents \cite{BT,PT1,PT2,PVZ}.

Real experiments have been done in the psychology literature in order to define the impact of the competence on a group decision-making, see \cite{BOLRRF,KD,MPNAS,HF} and the references therein.
Competence is generally associated to the predisposition to listen and give value to the other  opinions. The higher this quality, greater is the ability to value other opinions. Vice versa, a person unwilling to listen and dialogue is usually marked by not competent. An emergent phenomenon in group decision-making is called \emph{equality bias}, that is a misjudgement of personal competence of unskilled people during the exchange of informations, which goes hand in hand with the tendency of the most skilled individuals to underestimate their competence.

From the general structure introduced in the previous section we consider the evolution in $[0,T_f]$, $T_f>0$ of the following system of differential equations
\begin{equation}\begin{cases}\label{eq:micro}
\dot x_i = \displaystyle \sum_{j=1}^N \lambda(x_i,x_j) (x_j-x_i)+\lambda_B(x_i) z, \qquad i=1,\dots,N\\
\dot w_i= \dfrac{1}{N} \displaystyle\sum_{j=1}^N P(x_i,x_j;w_i,w_j)\left(w_j-w_i\right),
\end{cases}\end{equation}
where $z\in\RR^+$ is a the degree of competence achieved from the background at each interaction, having distribution $C(z)$ and bounded mean $m_B$
\begin{equation*}
\int_{\RR^+}C(z)dz=1,\qquad \int_{\RR^+}zC(z)dz=m_B.
\end{equation*}
Further, $\lambda_B(\cdot)$ quantifies the expertise gained from the background and $\lambda(\cdot,\cdot)$ weights the exchange of competence between individuals. A possible choice for the function $\lambda(\cdot,\cdot)$ is $\lambda(x_i,x_j)=\textrm{const.}>0$ if $x_i<x_j$ and $\lambda(x_i,x_j)=0$ elsewhere. In the above system we introduced the interaction function $0 \le P(w_i,w_j;x_i,x_j)\le 1$ depending on both the decisions and competence of the interacting agents.

More realistic models may be obtained by adding to \eqref{eq:micro} decision dependent noise terms modeling self-thinking processes and characterized by a function $D(x_i,w_i)\in[0,1]$ generally called local relevance of the diffusion for a given decision and competence.

 A possible choice for the interaction function is the following
\begin{equation}\label{eq:P}
P(w_i,w_j;x_i,x_j)=Q(w_i,w_j) R(x_i,x_j),
\end{equation}
where $0\le Q(\cdot,\cdot)\le 1$ is the compromise propensity and $0\le R(\cdot,\cdot)\le 1$ which takes into account the agents' competence. Let us assume $Q(w_i,w_j)\equiv 1$, we adopt the following notation for the square matrix $\mathcal R_N\in \textrm{Mat}_N(\RR^+)$
\begin{equation}\label{eq:R}
r_{ij}=R(x_i,x_j),\qquad \textrm{for all $i,j=1,\dots,N$.}
\end{equation}
 We further define the diagonal square matrix $\mathcal D_N\in \textrm{Mat}_N(\RR^+)$
 \begin{equation}\label{eq:D}
( \mathcal D_N)_{ij}=
 \begin{cases}
\sum_{j=1}^N r_{ij} & \textrm{if $i=j$}\\
0 & \textrm{if $i\ne j$,}
\end{cases} \end{equation}
for all $i,j=1,\dots,N$. Then we can rewrite \eqref{eq:micro} as follows
 \begin{equation*}\begin{split}
\dot w_i & = -\left(\dfrac{1}{N} \sum_{j=1}^N r_{ij}\right)w_i(t)+\dfrac{1}{N} \sum_{j=1}^N r_{ij}w_j(t),\\
&=-\frac{1}{N}\left[\mathcal D_N w(t)\right]_i+\frac{1}{N}\left[\mathcal R_N w(t)\right]_i,\\
&=-\dfrac{1}{N}\left[\mathcal L_N w(t)\right]_i.
\end{split}\end{equation*}
being
\be\label{eq:laplacian}
\mathcal L_N=\mathcal D_N-\mathcal R_N;
\ee
 where $\mathcal L_N$ is usually called Laplacian matrix of a graph.

\subsection{Collective decision-making under equality bias}
In the following we consider two main models of decision-making inspired by real experiments \cite{BOLRRF,MPNAS}. The first model takes into account the competence of individuals: at each interaction the prevailing decision coincides with the one of the system with maximum competence. We will refer to this model as maximum competence model (MC). In the present setting the MC model may be obtained by considering the Heaviside-type interaction function $R(x_i,x_j)=:R_{MC}(x_i,x_j)$
\begin{equation*}
R_{MC}(x,x_*)=
\begin{cases}
1 & x<x_* \\
1/2 & x=x_* \\
0 & x>x_*.
\end{cases}
\end{equation*}
The function $R_{MC}(\cdot,\cdot)$ may be approximated through a smoothed continuous version of the MC model (cMC)
\be\label{eq:R_cMC}
R_{cMC}(x_i,x_j) = \dfrac{1}{1+e^{c(x_i-x_j)}},
\ee
with $c>>1$.\\

In order to reproduce the cited equality bias we consider here a competence based interaction function $R(x_i,x_j)=:R_{EB}(x_i,x_j)$ with the following properties: if the competences $x_i$ and $x_j$ are very close together, i.e. in the homogeneous case, there are not appreciable changes in the dynamics of the model, while if $x_i$ and $x_j$ sensibly differ we have $R_{EB}(x_i,x_j)\simeq 1$. An example is given by the sigmoid function
\begin{equation}\label{eq:R_eqbias}
 R_{EB}(x_i,x_j)=\dfrac{1}{1+e^{-c(x_i-x_j)}},
\end{equation}
with $c>0$ a given constant. We depict in Figure \ref{fig:1} the functions $R_{cMC}(\cdot,\cdot)$ and $R_{EB}(\cdot,\cdot)$ defined in \eqref{eq:R_cMC}-\eqref{eq:R_eqbias} for several choices of the constant $c>0$. \\

Observe how, both in the EB and cMC cases, the element of the matrix $\mathcal R_N$ introduced in \eqref{eq:R} is such that
\begin{equation*}
r_{ij}= 1-r_{ji},\qquad i,j=1,\dots,N.
\end{equation*}
The problem to study the eigenvalues distribution of the matrix $\mathcal L_N$ is not in general an easy task, and it depends on the connectivity coefficients index of the model. Under suitable assumptions, it has been addressed in \cite{CS} and \cite{Shen}. In our case the Laplacian matrix is not symmetric and a strategy similar to \cite{CS} cannot be used.
For this reason, in the following, we will face the problem of the eigenvalues distribution under simplifying assumptions. \\

\begin{figure}[tb]
\centering
\includegraphics[scale=0.28]{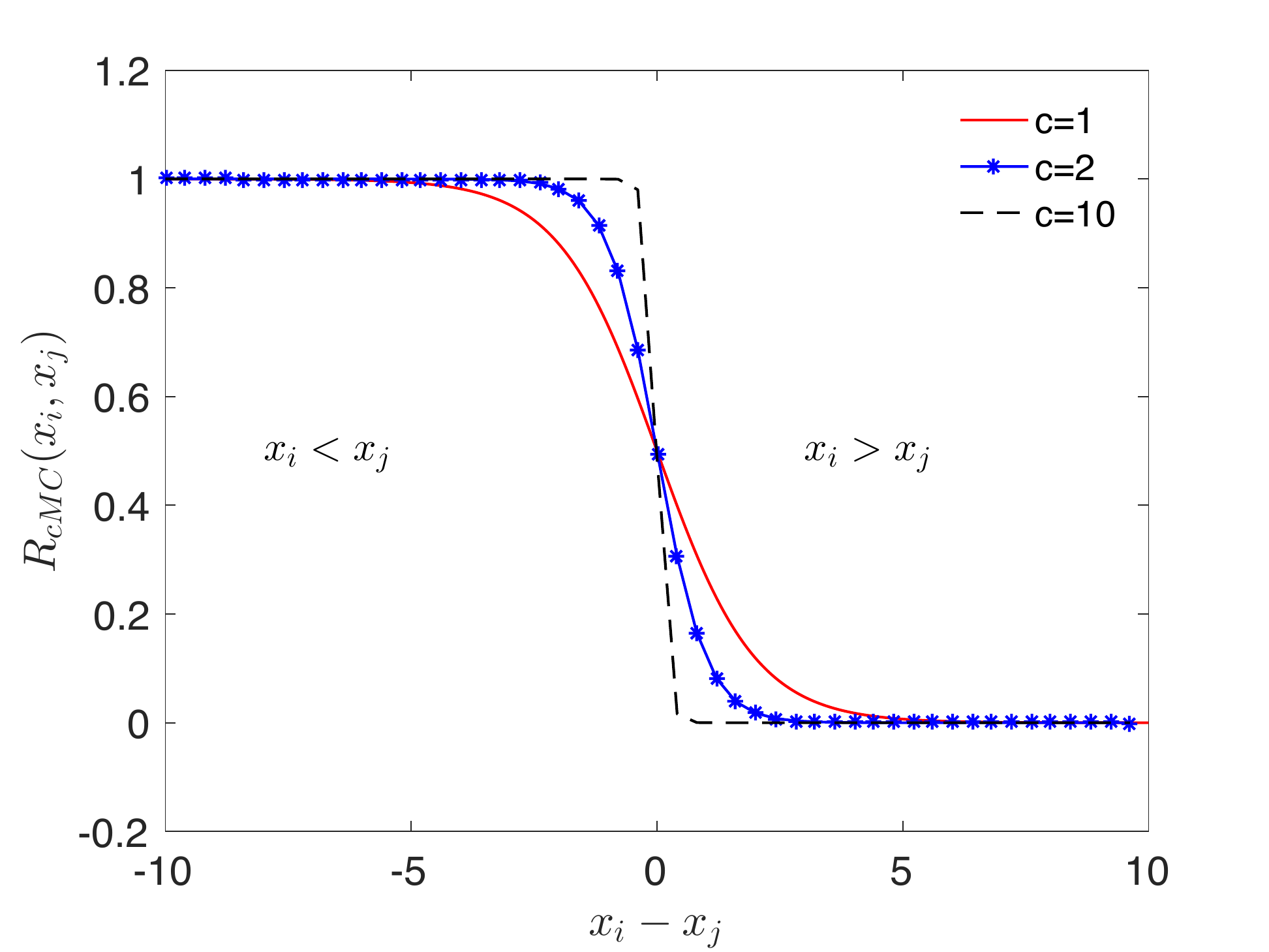}
\includegraphics[scale=0.28]{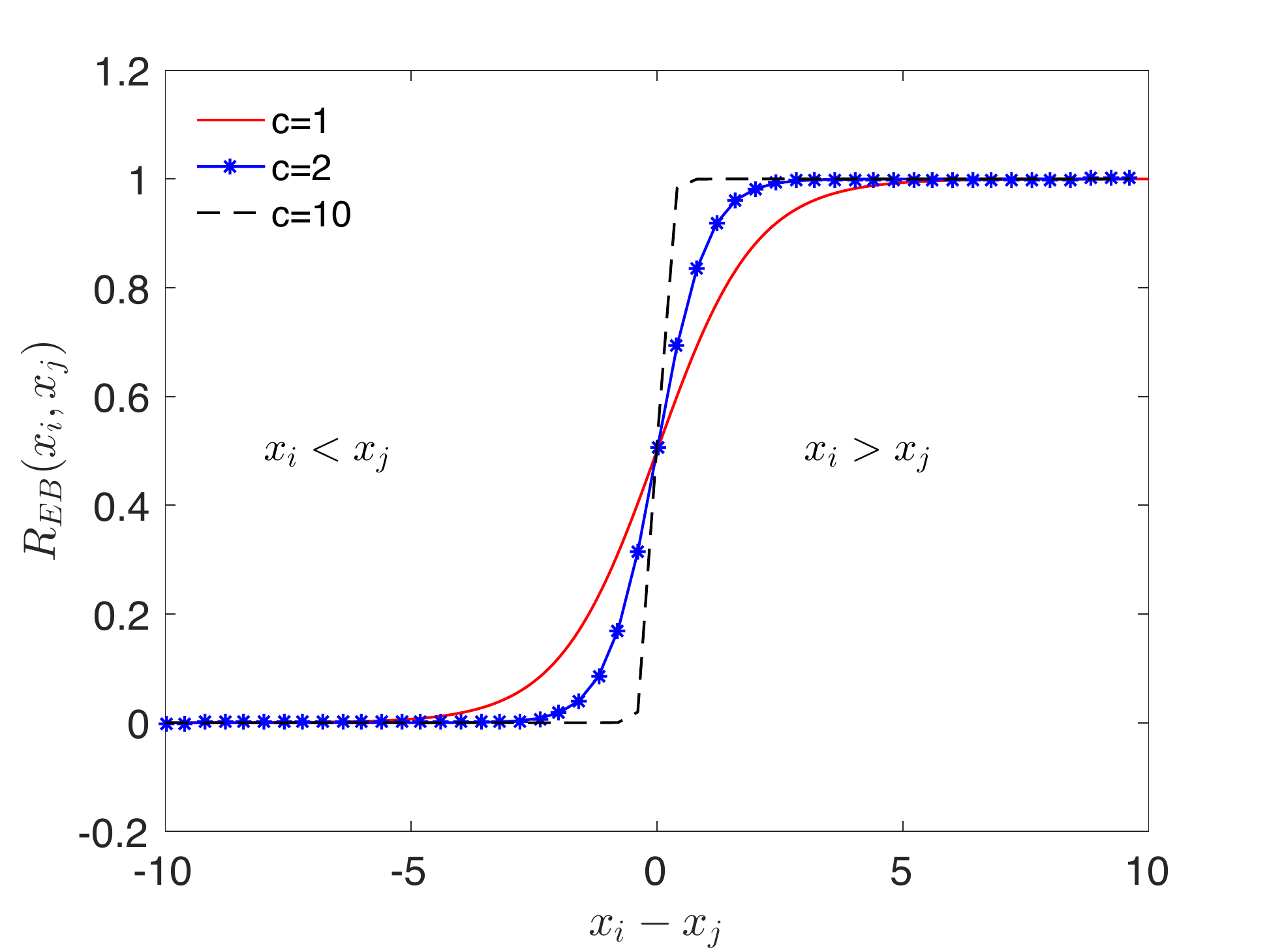}
\caption{We depict the functions $R_{cMC}(x_i,x_j)$ and $R_{EB}(x_i,x_j)$ defined in \eqref{eq:R_cMC} and \eqref{eq:R_eqbias} for several choices of the constant $c>0$. In the (cMC) case (left) in the half-plane $x_i>x_j$ the $i$th agent is scarcely influenced by the agent $j$th, the less competent one. On the other hand, in the (EB) case (right) we have the opposite situation in the half-plane $x_i>x_j$ and the less competent are not influenced by the more skilled agents. }
\label{fig:1}
\end{figure}

We introduce the concepts of collective decision \cite{G97,GZ}.
\begin{definition}[Collective decision]\label{def:1}
Let us consider a system of $N\in\mathbb N$ agents with competence and opinion $(x_i,w_i)_{i=1,\dots,N}$. We define the collective decision of the system the quantity
\[
\bar w = \dfrac{1}{N}\sum_{i=1}^N w_i.
\]
\end{definition}


\begin{definition}[Optimal collective decision]
\label{def:3}
Let $(x_i,w_i)_{i=1,\dots,N}$ be a system of $N\in\mathbb N$ interacting agents. A collective decision is said to be \emph{optimal} if
\begin{equation*}
\bar w = w_k,\qquad \textrm{such that}\qquad x_k=\max_{i=1,\dots,N}x_i.
\end{equation*}
\end{definition}
We have introduced a definition of optimal decision which is rather different from the one in \cite{PVZ}. In the cited work the optimal decision of the interacting system is suggested by external factors through and embedded in the dynamics through a self--propulsion term, whereas the introduced optimal decision depends on the maximal competence of the considered system of agents and is defined a priori as the decision of the most competent agents of the system. \\

In the rest of the paper we focus on two main situations described by the following assumptions:
\begin{assumption}\label{ass:1}
If the competence of the agents does not enter in the dynamics we have
\[
r_{ij}= r \qquad \textrm{for all}~i,j=1,\dots,N.
\]
\end{assumption}
\begin{assumption}\label{ass:2}
The system of agents is divided in two populations, competent and incompetent agents, belonging to the sets $S$ and $U$ respectively. The interaction function in the MC case reads
\[
r_{ij}=
\begin{cases}
0 & i\in S, j\in U\\
\dfrac{1}{2} & i,j\in S~ \textrm{or}~ i,j\in U\\
1 & i\in U, j\in S,
\end{cases}
\]
whereas it simplifies in the EB case as follows
\[
r_{ij}=
\begin{cases}
1 & i\in S, j\in U\\
\dfrac{1}{2} & i,j\in S~ \textrm{or}~ i,j\in U\\
0 & i\in U, j\in S.
\end{cases}
\]
\end{assumption}
The Assumption \ref{ass:1} describes the case in which the competences of individuals of the whole group does not depend on time and are very close together, i.e. $x_i\approx x_j$ for all $i,j=1,\dots,N$. We will analyze this simple case by assuming from $r=1/2$. On the other hand Assumption \ref{ass:2} define simplified interaction rules, which are coherent with $R_{MC}(\cdot,\cdot)$ and $R_{EB}(\cdot,\cdot)$ for $c\gg 1$. Observe how in this case the evolution of the competence variable does not permit the decoupling of the introduced decision dynamics and to derive explicit stationary decisions of the system. In the following we derive the general structure of the eigenvalues of the Laplacian matrix showing how the introduction of these interaction rules leads the agents toward an optimal or suboptimal collective decisions respectively in the MC and EB cases.

\section{Properties of the model}
\label{sec:simpl}
In this section we investigate the structure and properties of the matrices defined in the last section.

\begin{proposition}
\label{p:1}
Let us consider the matrices $\mathcal L_N=\mathcal D_N-\mathcal R_N$ defined in \eqref{eq:R}-\eqref{eq:D}. Then we have:
\begin{itemize}
  \item[(i)] The entries of $\mathcal L_N$ are given by
 \begin{equation}\label{eq:LN_hyp}
( \mathcal L_N)_{ij}=
\begin{cases}
(1-\delta_{i1})\left(i-1-\displaystyle\sum_{k=1}^{i-1}r_{ki}\right)+ (1-\delta_{iN}) \displaystyle\sum_{k=i+1}^N r_{ik} & \mbox{if}\ i = j \\
-r_{ij} & \mbox{if}\ i < j \\
r_{ji}-1 & \mbox{if}\ i > j
\end{cases}
\end{equation}
where
\[
\delta_{ij} =
\begin{cases}
0 & i \neq j,   \\
1 & i=j.
  \end{cases}
\]
is the Kronecker's delta function. Further, the expression of $\mathcal L_N$ at the may be written in terms of $\mathcal L_{N-1}$ as follows
\begin{equation}
\label{eq:15}
\mathcal L_N=\left(
        \begin{matrix}
         \mathcal L_{N-1}+\mathcal H_{N-1} & -\overrightarrow{h}_{N-1}^T \\
          \overrightarrow{h}_{N-1}-\overrightarrow{1}_{N-1} & N-1-\sum_{i=1}^{N-1} r_{i,N} \\
        \end{matrix}
      \right)
\end{equation}
where $\mathcal L_1=(0)$ and we introduced the terms $\overrightarrow{h}_{N-1}=[r_{iN}]_{i=1,\dots,N-1}$, $\overrightarrow{1}_{N-1}=[\underbrace{1,\dots,1}_{N-1}]$ and the diagonal matrix $\mathcal H_{N-1}\in \textrm{Mat}_{N-1}([0,1])$ defined as
\be\label{eq:HN1}
(\mathcal H_{N-1})_{ij} =
\begin{cases}
r_{iN} & \textrm{if $i=j$},\\
0 & \textrm{otherwise}\, .
\end{cases}
\ee
  \item[(ii)] The matrix $\mathcal L_N$ is singular.
  \item[(iii)] For $N\geq 2$, $\operatorname{tr}(\mathcal L_N)=\frac{N(N-1)}{2}$.
\end{itemize}
\end{proposition}

\begin{proof}
\begin{itemize}
\item[(i)] By induction, in the case $N=2$ we have:
\begin{equation*}
\mathcal L_2=
\left(
\begin{matrix}
        r_{11}+r_{12} & 0 \\
        0                       & r_{21}+r_{22}
\end{matrix}
\right)-\left(
\begin{matrix}
        r_{11} & r_{12} \\
        r_{21} & r_{22} \\
\end{matrix}
    \right)=\left(
      \begin{matrix}
        r_{12} & -r_{12} \\
        r_{12}-1 & 1-r_{12} \\
      \end{matrix}
    \right),
\end{equation*}
that is \eqref{eq:LN_hyp} in the case $N=2$. We assume true (\ref{eq:15}) for  any$N\in\mathbb N$, $N>2$, therefore we have
\begin{equation*}
\mathcal L_{N+1}=\mathcal D_{N+1}-\mathcal R_{N+1},
\end{equation*}
that is
\begin{equation}\begin{split}\label{lem1:mat}
\mathcal L_{N+1}=&
\left(
    \begin{array}{cccc}
      r_{1,1}+\dots+r_{1,N+1} &  & 0 \\
          & \ddots & \\
      0 &             & r_{N+1,1}+\dots+r_{N+1,N+1} \\
    \end{array}
  \right)\\
&-\left(
     \begin{array}{cccc}
       r_{1,1} & r_{1,2} & \dots & r_{1,N+1} \\
       r_{2,1} & r_{2,2} & \dots & r_{2,N+1} \\
       \vdots & \vdots & \vdots & \vdots \\
       r_{N+1,1} & r_{N+1,2} & \dots & r_{N+1,N+1} \\
     \end{array}
   \right).
\end{split}\end{equation}
Being $-r_{i,j}=r_{j,i}-1$ for $i>j$ we have that the $(N+1)$th row of \eqref{lem1:mat} is given by
\begin{equation*}
\left(\overrightarrow{h}_{N}-\overrightarrow{1}_{N}, N-\sum_{i=1}^{N} r_{i,N+1}\right),
\end{equation*}
while we can write the $(N+1)$th column as,
\begin{equation*}
\left(\begin{array}{c}
    -\overrightarrow{h}_{N}^T \\
    N-\sum_{i=1}^{N} r_{i,N+1}
  \end{array}\right).
\end{equation*}
Hence, if we define
\begin{equation*}
(\mathcal H_{N+1})_{ij} =
\begin{cases}
r_{i,N+1} & \textrm{if $i=j$},\\
0 & \textrm{otherwise},
\end{cases}
\end{equation*}
we have the first point.

\item[(ii)] It follows form the fact that the vector $(1,\dots,1)$  lies in the kernel of the matrix $\mathcal L_N$.

\item[(iii)] We have to show that
\begin{equation*}
\sum_{i=1}^N (1-\delta_{i1})\left(i-1-\sum_{k=1}^{i-1}r_{ki}\right)+ \sum_{i=1}^N (1-\delta_{iN}) \sum_{k=i+1}^N r_{ik}=\frac{N(N-1)}{2},
\end{equation*}
that is
\be\label{eq:prop1_1}
\sum_{i=2}^N \left(i-1-\sum_{k=1}^{i-1}r_{ki}\right)+ \sum_{i=1}^{N-1} \sum_{k=i+1}^N r_{ik}=\frac{N(N-1)}{2},
\ee
which is true in the case $N=2$.
Therefore, we prove by mathematical induction that equation \eqref{eq:prop1_1} holds for all $N>2$.
Let us define the following objects
\begin{equation*}\begin{split}
P_N&=\sum_{i=2}^N \left(i-1-\sum_{k=1}^{i-1}r_{ki}\right), \\
Q_N&=\sum_{i=1}^{N-1} \sum_{k=i+1}^N r_{ik}.
\end{split}\end{equation*}
It follows that
\begin{equation*}
\begin{split}
P_{N+1}&=P_N+N-\sum_{k=1}^{N}r_{k(N+1)}, \\
Q_{N+1}&=Q_N+\sum_{k=1}^{N}r_{k(N+1)},
\end{split}
\end{equation*}
thus
\begin{equation*}
P_{N+1}+Q_{N+1}=P_N+Q_N+N=\frac{N(N-1)}{2}+N=\frac{N(N+1)}{2},
\end{equation*}
which completes the proof.
\end{itemize}
\end{proof}

\begin{lemma}
\label{l:1}
Let us consider the matrix $\mathcal R_N$ under the Assumption \ref{ass:1}, i.e. $r_{ij}=r\in[0,1]$ for each $i,j=1,\dots N$, $N\ge 2$. Denoting with $\lambda_2^{N-1},\dots, \lambda_{N-1}^{N-1}$ the non-zero eigenvalues of $\mathcal L_{N-1}$, the expression of $\mathcal L_{N}$ is
\begin{equation}
\label{eq:17}
\mathcal L_N=\left(
    \begin{matrix}
      (N-1)r & -r & -r & \ldots & -r \\
      r-1 & \lambda_2^{N-1} & -r & \ldots & -r \\
      \vdots & \dots & \ddots & \ldots & \vdots \\
      r-1 & \ldots & r-1 & \lambda_{N-1}^{N-1} & -r \\
      r-1 & r-1 & \dots & r-1 & (N-1)(1-r) \\
    \end{matrix}
  \right)
\end{equation}
with eigenvalues
\begin{equation}
\label{eq:16}
\lambda_1^N=0, \ \ \lambda_i^N=(i-1)(1-r)+(N+1-i)r, \ \ i=2,\dots,N
\end{equation}
\end{lemma}

\begin{proof}
In the case $N=2$ we have
\begin{equation*}
\lambda_1^2=0,\qquad \lambda_2^2=1
\end{equation*}
which are eigenvalues of the matrix
\begin{equation*}
\mathcal L_2 =\left(
    \begin{matrix}
      r & -r \\
      r-1 & 1-r \\
    \end{matrix}
  \right)
\end{equation*}
We proceed by induction, assume the statement is true for a generic integer $N-1$; by equation (\ref{eq:15}) in Proposition \ref{p:1} we have:
\begin{equation*}
\begin{split}
\mathcal L_N&=\left(
        \begin{matrix}
          \mathcal L_{N-1}+ \mathcal H_{N-1} & -\overrightarrow{h}_{N-1}^T \\
          \overrightarrow{h}_{N-1}-\overrightarrow{1}_{N-1} & (N-1)(1-r) \\
        \end{matrix}
      \right)\\
      &=\left(
        \begin{matrix}
          \mathcal L_{N-1}+r Id_{N-1} & -\overrightarrow{h}_{N-1}^T \\
          \overrightarrow{h}_{N-1}-\overrightarrow{1}_{N-1} & (N-1)(1-r) \\
        \end{matrix}
      \right)
\end{split}
\end{equation*}
where $\overrightarrow{h}_{N-1} = [\underbrace{r,\dots,r}_{N-1}]$ and $Id_{N-1}$ is the identity matrix of size $N-1$. Therefore, $\mathcal L_N$ assumes the form given in \eqref{eq:17} for each $N\ge 2$. Let us consider now $\lambda_i^N$ as in \eqref{eq:16}, we prove that
\begin{equation*}
\det\left(\mathcal L_N-\lambda_i^N Id_N\right)=0, \qquad \textrm{for each $ i=2,\dots, N$.}
\end{equation*}
If $i=2$, the first two rows of $\mathcal L_N-\lambda_2^N Id_N$ are
\be\label{eq:2rows}
\begin{matrix}
    (N-1)r-\lambda_2^N & -r & -r & \ldots & -r \\
    r-1 & \lambda_2^{N-1}-\lambda_2^N & -r & \ldots & -r
  \end{matrix}
\ee
with $\lambda_2^N=1-r+(N-1)r$, and the two rows in \eqref{eq:2rows} are both equal to the array
\begin{equation*}
r-1,-r,\dots,-r
\end{equation*}
For $2<i+1<N$, we consider the $i$th and $(i+1)$th rows of $\mathcal L_N-\lambda_i^N Id_N$, given by
\begin{equation*}
\begin{matrix}
    r-1 & \dots & r-1 & \lambda_i^{N-1}-\lambda_{i+1}^N  & -r & -r & \dots & -r \\
    r-1 & \dots & r-1 & r-1 & \lambda_{i+1}^{N-1}-\lambda_{i+1}^N  & -r & \dots & -r
  \end{matrix}
\end{equation*}
where
\begin{equation*}
\lambda_{i+1}^{N-1}-\lambda_{i+1}^N=-r, \qquad \lambda_i^{N-1}-\lambda_{i+1}^N=r-1.
\end{equation*}
Thus the $i$th and $(i+1)$th rows are linearly dependent. Finally we consider the case $i=N$, we observe that the last two rows of $\mathcal L_N -\lambda_N^N Id_N$
\begin{equation*}
\begin{matrix}
    r-1 & \dots & r-1 & \lambda_{N-1}^{N-1}-\lambda_{N}^N  & -r  \\
    r-1 & \dots & r-1 & r-1 & \lambda_{N}^{N-1}-\lambda_{N}^N
  \end{matrix}
\end{equation*}
are equal, in fact $\lambda_{N-1}^{N-1}-\lambda_{N}^N=r-1$ and $\lambda_{N}^{N-1}-\lambda_{N}^N=-r$. We have proven that $\lambda_i^N$, for $i=2,\dots,N$ defined in \eqref{eq:16} are solutions of characteristic polynomial associated to the matrix $\mathcal L_N$.
\end{proof}

From the Lemma \ref{l:1} we can easily see that the eigenvalues \eqref{eq:16} of $\mathcal L_{N}$ are real and positive. Further it is now easy to show that the Laplacian matrix $\mathcal L_N$ defined in \eqref{eq:laplacian}, in the case described by Assumption \ref{ass:1}, assumes the form given by the following result.
\begin{corollary}
\label{c:1}
Let us consider Assumption \ref{ass:1} with all entries of the matrix $\mathcal R_N$ such that $r_{ij}=\frac{1}{2}$ for each $i,j=1,\dots N$. The Laplacian matrix $\mathcal L_N$ is the following:
\be \label{eq:cor}
\mathcal L_{N}=\left(
    \begin{matrix}
      \frac{N-1}{2} & -\frac{1}{2} & -\frac{1}{2} & \ldots & -\frac{1}{2} \\
      -\frac{1}{2} & \frac{N-1}{2} & -\frac{1}{2} & \ldots & -\frac{1}{2} \\
      \vdots & \dots & \ddots & \ldots & \vdots \\
      -\frac{1}{2} & \ldots & -\frac{1}{2} & \frac{N-1}{2} & -\frac{1}{2} \\
      -\frac{1}{2} & -\frac{1}{2} & \dots & -\frac{1}{2} & \frac{N-1}{2} \\
    \end{matrix}
  \right),
\ee
with eigenvalues
\begin{equation*}
\lambda_1^N=0,\qquad \lambda_i^N=\dfrac{N}{2}, \qquad i=2,\dots,N.
\end{equation*}
\end{corollary}

In the following we utilize the notations $\mathcal L_N^{MC}$ and $\mathcal L_N^{EB}$ for the Laplacian matrix under the maximum competence and equality bias case respectively as in Assumption \ref{ass:2}. The Laplacian matrix under the Assumption \ref{ass:1} will be denoted with $\mathcal L_N^{(1)}$.

In the following example we establish the structures of the matrices $\mathcal L_N^{MC}$ and $\mathcal L_N^{EB}$. More exhaustive results are proposed in Lemma \ref{l:2}.
\begin{example}
\label{ex:1a}


In this example we consider four agents where $i=1,2$ have high competence and the agents $i=3,4$ have no competence (assuming $r_{12}=r_{21}=r_{34}=r_{43}=\frac{1}{2}$). We denote it with \textbf{Case a}. Afterward, we consider the \textbf{Case b}, where $i=1$ has high competence and the agents $i=2,3,4$ have no competence. All the computations refer to the case of wide competence gap for the EB model, or equivalently to the case $c>>1$.



\textbf{Case a (MC model)}. Let
\begin{equation*}
r_{13}=r_{14}=r_{23}=r_{24}=0, \qquad r_{31}=r_{41}=r_{32}=r_{42}=1
\end{equation*}
and
\begin{equation*}
r_{12}=r_{21}=r_{34}=r_{43}=\frac{1}{2}\, .
\end{equation*}
Hence
\renewcommand{\arraystretch}{2}
\begin{equation}\label{eq:21}\newcommand*{\temp}{\multicolumn{1}{r|}{}}
\mathcal L_{4}^{MC}=\left(
    \begin{matrix}
      \frac{1}{2} & -\frac{1}{2}&\temp & 0 & 0 \\
      -\frac{1}{2} & \frac{1}{2}&\temp & 0 & 0 \\ \cline{1-5}
      -1 & -1&\temp & \frac{5}{2} & -\frac{1}{2} \\
      -1 & -1&\temp & -\frac{1}{2} & \frac{5}{2} \\
    \end{matrix}
  \right)
\end{equation}
which is a matrix of triangular block form. The characteristic polynomial of $\mathcal L_4^{MC}$ is
\begin{equation*}
p\left(\mathcal L_{4}^{MC},\lambda\right)=p\left(\mathcal L_2^{(1)},\lambda\right)\cdot p\left(\mathcal L_2^{(1)}+2Id_2,\lambda\right) \, ,
\end{equation*}
and so the eigenvalues are $\{0,1,2,3\}$. The possibility of calculate the eigenvalues of $\mathcal L_{4}^{MC}$ from those of $\mathcal L_2^{(1)}$, is due to a decoupling of two effects: the effect imposed on the system by highly skilled agents, and the effect of the less competent agents. Note that, since the first two agents are those most competent, the upper left block in \eqref{eq:21} is related to highly skilled agents ($i=1,2$) while the lower right block is due to less competent agents ($i=3,4$).


\textbf{Case a (EB model)}.
\begin{equation*}
r_{13}=r_{14}=r_{23}=r_{24}=1, \qquad r_{31}=r_{41}=r_{32}=r_{42}=0
\end{equation*}
and
\begin{equation*}
r_{12}=r_{21}=r_{34}=r_{43}=\frac{1}{2}.
\end{equation*}
Hence
\renewcommand{\arraystretch}{2}
\begin{equation}\label{eq:21}\newcommand*{\temp}{\multicolumn{1}{r|}{}}
\mathcal L_{4}^{EB}=\left(
    \begin{matrix}
      \frac{5}{2} & -\frac{1}{2}&\temp & -1 & -1 \\
      -\frac{1}{2} & \frac{5}{2}&\temp & -1 & -1 \\ \cline{1-5}
      0 & 0&\temp & \frac{1}{2} & -\frac{1}{2} \\
      0 & 0&\temp & -\frac{1}{2} & \frac{1}{2} \\
    \end{matrix}
  \right)
\end{equation}
which is a matrix of triangular block form. The characteristic polynomial of $\mathcal L_4^{EB}$ is
\begin{equation*}
p\left(\mathcal L_{4}^{EB},\lambda\right)=p\left(\mathcal L_2^{(1)}+2Id_2,\lambda\right)\cdot p\left(\mathcal L_2^{(1)},\lambda\right)\, ,
\end{equation*}
and so the eigenvalues are still $\{0,1,2,3\}$. As for MC model, we can observe again the decoupling of the effect imposed on the system by highly skilled agents from those of the less competent agents.

\textbf{Case b (MC model)}.
\begin{equation*}
  r_{12}=r_{13}=r_{14}=0, \qquad r_{21}=r_{31}=r_{41}=1
\end{equation*}
and
\begin{equation*}
  r_{11}=r_{22}=r_{23}=r_{24}=r_{32}=r_{33}=r_{34}=r_{42}=r_{43}=r_{44}=\frac{1}{2}\, .
\end{equation*}
Then
\renewcommand{\arraystretch}{2}
\begin{equation*}\newcommand*{\temp}{\multicolumn{1}{r|}{}}
\mathcal L_{4}^{MC}=\left(
    \begin{matrix}
      0 &\temp& 0& & 0 & 0 \\ \cline{1-6}
      -1 &\temp& 2& & -\frac12 & -\frac12 \\
      -1 &\temp &-\frac12& & 2 & -\frac{1}{2} \\
      -1 &\temp &-\frac12& & -\frac{1}{2} & 2 \\
    \end{matrix}
  \right)\, .
\end{equation*}
The characteristic polynomial of $\mathcal L_4^{MC}$ is
\begin{equation*}
p\left(\mathcal L_{4}^{MC},\lambda\right)=\lambda \, p\left(\mathcal L_3^{(1)}+Id_3,\lambda\right)\, ,
\end{equation*}
whose eigenvalues are $\{0,1,\frac52,\frac52\}$.

\textbf{Case b (EB model)}.
\begin{equation*}
  r_{12}=r_{13}=r_{14}=1, \qquad r_{21}=r_{31}=r_{41}=0
\end{equation*}
and
\begin{equation*}
  r_{11}=r_{22}=r_{23}=r_{24}=r_{32}=r_{33}=r_{34}=r_{42}=r_{43}=r_{44}=\frac{1}{2}\, .
\end{equation*}
Hence
\renewcommand{\arraystretch}{2}
\begin{equation*}\newcommand*{\temp}{\multicolumn{1}{r|}{}}
\mathcal L_{4}^{EB}=\left(
    \begin{matrix}
      3 &\temp& -1& & -1 & -1 \\ \cline{1-6}
      0 &\temp& 1& & -\frac12 & -\frac12 \\
      0 &\temp &-\frac12& & 1 & -\frac{1}{2} \\
      0 &\temp &-\frac12& & -\frac{1}{2} & 1 \\
    \end{matrix}
  \right)\, .
\end{equation*}
The characteristic polynomial of $\mathcal L_4^{EB}$ is
\begin{equation*}
p\left(\mathcal L_{4}^{EB},\lambda\right)=(\lambda-3) \, p\left(\mathcal L_3^{(1)},\lambda\right)\, ,
\end{equation*}
whose eigenvalues are $\{0,\frac{3}{2},\frac{3}{2},3\}$.
\end{example}
In the following Lemma we generalize this approach denoting with $N_1$ the number of incompetent agents which may vary in time. Besides, we recall the notation introduced with Proposition \ref{p:1}: $\overrightarrow{h}_{N-1}=[r_{1,N},\dots,r_{N-1,N}]$, $\overrightarrow{1}_{N-1}=[\underbrace{1,\dots,1}_{N-1}]$, $\overrightarrow{0}_{N-1}=[\underbrace{0,\dots,0}_{N-1}]$), and $\mathcal H_{N-1}$ is the diagonal matrix defined in \eqref{eq:HN1}.

From the Lemma \ref{l:1} we can easily see that the eigenvalues \eqref{eq:16} of $\mathcal L_{N}$ are real and positive. Further it is now easy to show that the Laplacian matrix $\mathcal L_N$ defined in \eqref{eq:laplacian}, in the case described by Assumption \ref{ass:2}, assumes the form given by the following result.

\begin{lemma}
\label{l:2}
Under Assumption \ref{ass:2} let us consider a system of $N+N_1\in \mathbb N$, $N\ge1$, $N_1\ge1$ interacting agents such that $S=\{1,\dots,N\}$ and $U=\{N+1,\dots,N+N_1\}$. We define the following rectangular matrices
$
\mathcal J \in \textrm{Mat}_{N-N_1,N_1}(\{1\}),
\mathcal O \in\textrm{Mat}_{N_1,N-N_1}(\{0\})
$, i.e.
\begin{equation*}
\begin{split}
(\mathcal J)_{ij} = 1,\qquad \textrm{for all $1\le i \le N-N_1, 1\le j \le N_1$}, \\
(\mathcal O)_{ij} = 0,\qquad \textrm{for all $1\le i \le N_1, 1\le j \le N-N_1$}. \\
\end{split}
\end{equation*}
Then, in the case $c>>1$, we have the following claims:
\begin{enumerate}
  \item[(i)] The Laplacian matrix for the MC model is given by
  \begin{equation}\label{eq:22}\newcommand*{\temp}{\multicolumn{1}{r|}{}}
\mathcal L_{N+N_1}^{MC}=\left(
          \begin{array}{ccc}
            \mathcal L_{N}^{(1)}\quad & \temp &\quad \mathcal O \\ \cline{1-3}
            - \mathcal J & \temp &\quad  \mathcal L_{N_1}^{(1)}+N Id_{N_1} \\
          \end{array}
        \right)
\end{equation}
its characteristic polynomial is given by
\begin{equation}\label{eq:23}
p\left(\mathcal L_{N+N_1}^{MC},\lambda\right)=p\left(\mathcal L_{N}^{(1)},\lambda\right)\cdot p\left(\mathcal L_{N_1}^{(1)}+N\, Id_{N_1},\lambda\right) \, .
\end{equation}
with eigenvalues:
\begin{equation*}\begin{split}
&\lambda_1=0, \qquad \lambda_{2}=\dots=\lambda_N=\frac{N}{2},\\
&\lambda_{N+1}=N,\qquad \lambda_{N+2}=\dots=\lambda_{N+N_1}=\frac{N_1}{2}+N,
\end{split}\end{equation*}
  \item[(ii)] The Laplacian matrix for the EB model is given by
\begin{equation}\label{eq:22b}\newcommand*{\temp}{\multicolumn{1}{r|}{}}
\mathcal L_{N+N_1}^{EB}=\left(
          \begin{array}{ccc}
            \mathcal L_{N}^{(1)}+N_1 Id_{N}\quad & \temp &\quad - \mathcal J \\ \cline{1-3}
            \mathcal O & \temp &\quad  \mathcal L_{N_1}^{(1)} \\
          \end{array}
        \right)
\end{equation}
its characteristic polynomial is given by
\begin{equation}\label{eq:23b}
p\left(\mathcal L_{N+N_1}^{EB},\lambda\right)=p\left(\mathcal L_{N}^{(1)}+N_1\, Id_{N},\lambda\right)\cdot p\left(\mathcal L_{N_1}^{(1)},\lambda\right)
\end{equation}
with eigenvalues
\begin{equation*}\begin{split}
&\lambda_{1}=N_1,\qquad \lambda_2=\dots=\lambda_{N}= \frac{N}{2}+N_1,\\
&\lambda_{N+1}=\dots=\lambda_{N+N_1-1}= \frac{N_1}{2},\qquad \lambda_{N+N_1}=0.
\end{split}\end{equation*}
\end{enumerate}
\end{lemma}
\begin{proof}
\begin{itemize}
\item[(i)] We proceed by induction. Let us consider $N_1=1$, therefore $S=\left\{i=1, \dots ,N\right\}$ and $U=\left\{N+1\right\}$.
From Proposition \ref{p:1}, eq. \eqref{eq:15}, we have
\renewcommand{\arraystretch}{2}
\begin{equation*}
\newcommand*{\temp}{\multicolumn{1}{r|}{}}
\mathcal L_{N+1}^{MC}=\left(
	\begin{matrix}
         \mathcal L_{N}^{(1)}+\mathcal H_{N} & \temp & - \overrightarrow{h}_N^T \\ \cline{1-3}
            \overrightarrow{h}_N-\overrightarrow{1}_N & \temp & N-\sum_{i=1}^{N} r_{i\ N+1} \\
          \end{matrix}
          \right)
\end{equation*}
where $r_{i,\, N+1}=0$,  for each $i=1,\dots, N$. Hence $\mathcal L_{N+1}^{MC}$ assumes the following form
\renewcommand{\arraystretch}{2}
\begin{equation*}
\newcommand*{\temp}{\multicolumn{1}{r|}{}}
 \mathcal L_{N+1}^{MC}=\left(
 \begin{matrix}
            \mathcal L_{N}^{(1)}+Id_{N} & \temp & \overrightarrow{0}_N^T \\ \cline{1-3}
            - \overrightarrow{1}_N & \temp & N \\
          \end{matrix}
\right)
\end{equation*}
The first step has been showed. Let us assume that the results in \eqref{eq:22} and \eqref{eq:23} hold. We add another incompetent agent such that: $S=\left\{1,\dots,N\,\right\}$, $U=\{N+1,\dots,N+N_1+1\}$. From Proposition \ref{p:1} we have
\renewcommand{\arraystretch}{2}
\begin{equation*}
\newcommand*{\temp}{\multicolumn{1}{r|}{}}
\mathcal L_{N+N_1+1}^{MC}=\\
\left(
          \begin{array}{ccc}
            \left(
          \begin{array}{ccc}
            \mathcal L_{N}^{(1)} & \temp & \mathcal O \\ \cline{1-3}
            -\mathcal J & \temp & \mathcal L_{N_1}^{(1)}+ N Id_{N_1} \\
          \end{array}
        \right)+\mathcal H_{N+N_1}& \temp & - \overrightarrow{h}_{N+N_1}^T \\ \cline{1-3}
            \overrightarrow{h}_{N+N_1}-\overrightarrow{1}_{N+N_1} & \temp & N+N_1-\sum_{i=1}^{N+N_1} r_{i,\, N+N_1+1} \\
          \end{array}
        \right),
\end{equation*}
where
\begin{align*}
  r_{1,N+N_1+1}= & \dots=r_{N,N+N_1+1}=0 \\
  r_{N+1,N+N_1+1}= & \dots=r_{N+N_1,N+N_1+1}=\frac{1}{2} \, ,
\end{align*}
and so
\renewcommand{\arraystretch}{2}
\begin{equation*}\begin{split}
\newcommand*{\temp}{\multicolumn{1}{r|}{}}
\left(\begin{array}{ccc}
            \mathcal L_{N}^{(1)}& \temp & \mathcal O \\ \cline{1-3}
            -\mathcal J & \temp & \mathcal L_{N_1}^{(1)}+ N Id_{N_1} \\
          \end{array}\right)+\mathcal H_{N+N_1}
=\left(\begin{array}{ccc}
            \mathcal L_{N}^{(1)}& \temp & \mathcal O \\ \cline{1-3}
            -\mathcal J & \temp & \mathcal L_{N_1}^{(1)}+\left(N+ \frac{1}{2}\right) Id_{N_1} \\
          \end{array}\right).
\end{split} \end{equation*}
Rows $-\overrightarrow{h}_{N+N_1}$ and $\overrightarrow{h}_{N+N_1}-\overrightarrow 1_{N+N_1}$ are respectively given by
$$\underbrace{0\, , \, \dots \, , \, 0}_{N},\underbrace{-\frac{1}{2}\, , \, \dots \, , \, -\frac{1}{2}}_{N_1}$$
$$\underbrace{-1\, ,\, \dots \, , \,-1}_{N},\underbrace{-\frac{1}{2}\, , \, \dots \, , \, -\frac{1}{2}}_{N_1}\, ,$$
while
$$N+N_1-\sum_{i=1}^{N+N_1} r_{i,\, N+N_1+1}=N+\frac{N_1}{2}\, .$$
From Corollary \ref{c:1}, the main diagonal of $L_{N_1}^{(1)}$ contains all $\frac{N_1-1}{2}$, hence the inductive step is proved.

\item[(ii)] We proceed by induction. Let us consider $N_1=1$, therefore $S=\left\{1, \dots ,N\right\}$ and $U=\left\{N+1\right\}$.
From Proposition \ref{p:1}, eq. \eqref{eq:15}, we have
\renewcommand{\arraystretch}{2}
\begin{equation*}
\newcommand*{\temp}{\multicolumn{1}{r|}{}}
\mathcal L_{N+1}^{EB}=\left(
	\begin{matrix}
         \mathcal L_{N}^{(1)}+\mathcal H_{N} & \temp & - \overrightarrow{h}_N^T \\ \cline{1-3}
            \overrightarrow{h}_N-\overrightarrow{1}_N & \temp & N-\sum_{i=1}^{N} r_{i\ N+1} \\
          \end{matrix}
          \right)
\end{equation*}
where $r_{i\ N+1}=1$,  for each $i=1,\dots, N$. Hence $\mathcal L_{N+1}^{EB}$ assumes the following form
\renewcommand{\arraystretch}{2}
\begin{equation*}
\newcommand*{\temp}{\multicolumn{1}{r|}{}}
 \mathcal L_{N+1}^{EB}=\left(
 \begin{matrix}
            \mathcal L_{N}^{(1)}+Id_{N} & \temp & - \overrightarrow{1}_N^T \\ \cline{1-3}
            \overrightarrow{0}_N & \temp & 0 \\
          \end{matrix}
\right)
\end{equation*}
The base step is achieved. Suppose the result (\ref{eq:22}) and (\ref{eq:23}) holds. We add another incompetent agent: $S=\left\{1,\dots, N\,\right\}$, $U=\{N+1,\dots , N+N_1+1\}$. We have, from Proposition \ref{p:1}, eq. (\ref{eq:15}), $\mathcal L_{N+N_1+1}^{EB}=$
\renewcommand{\arraystretch}{2}
\begin{equation*}\newcommand*{\temp}{\multicolumn{1}{r|}{}}
\left(
          \begin{array}{ccc}
            \left(
          \begin{array}{ccc}
           \mathcal L_{N}^{(1)}+N_1\, I_{N} & \temp & - \mathcal J_{N,N_1} \\ \cline{1-3}
           \mathcal O_{N_1,N} & \temp & \mathcal L_{N_1}^{(1)} \\
          \end{array}
        \right)+\mathcal H_{N+N_1}& \temp & - \overrightarrow{h}_{N+N_1}^T \\ \cline{1-3}
            \overrightarrow{h}_{N+N_1}-\overrightarrow{1}_{N+N_1} & \temp & N+N_1-\sum_{i=1}^{N+N_1} r_{i,N+N_1+1} \\
          \end{array}
        \right) \, ,
\end{equation*}
where
\begin{align*}
  r_{1,N+N_1+1}= & \dots=r_{N,N+N_1+1}=1 \\
  r_{N+1,N+N_1+1}= & \dots=r_{N+N_1,N+N_1+1}=\frac{1}{2}.
\end{align*}
Therefore we have
\renewcommand{\arraystretch}{2}
\begin{equation*}\begin{split}
\newcommand*{\temp}{\multicolumn{1}{r|}{}}
\left(\begin{array}{ccc}
            \mathcal L_{N}^{(1)}+N_1\, I_{N} & \temp & - \mathcal J_{N,N_1} \\ \cline{1-3}
            \mathcal O_{N_1,N} & \temp & \mathcal L_{N_1}^{(1)} \\
          \end{array}\right)+\mathcal H_{N+N_1}=
\end{split}\end{equation*}
\begin{equation*}\begin{split}
\renewcommand{\arraystretch}{2}
\newcommand*{\temp}{\multicolumn{1}{r|}{}}
=\left(\begin{array}{ccc}
           \mathcal L_{N}^{(1)}+(N_1+1)\, I_{N} & \temp & - \mathcal J_{N,N_1} \\ \cline{1-3}
           \mathcal O_{N_1,N} & \temp & \mathcal L_{N_1}^{(1)}+\frac{1}{2} I_{N_1} \\
          \end{array}\right)\, .
\end{split}\end{equation*}
Rows $-\overrightarrow{h}_{N+N_1}$ and $\overrightarrow{h}_{N+N_1}-\overrightarrow 1_{N+N_1}$ are, respectively:
$$\underbrace{-1\, ,\, \dots \, , \,-1}_{N},\underbrace{-\frac{1}{2}\, , \, \dots \, , \, -\frac{1}{2}}_{N_1}$$
$$\underbrace{0\, , \, \dots \, , \, 0}_{N},\underbrace{-\frac{1}{2}\, , \, \dots \, , \, -\frac{1}{2}}_{N_1} \, ,$$
while
$$N+N_1-\sum_{i=1}^{N+N_1} r_{i,N+N_1+1}=\frac{N_1}{2}\, .$$
From Corollary \ref{c:1}, the main diagonal of $\mathcal L_{N_1}^{(1)}$ contains all $\frac{N_1-1}{2}$, hence
\begin{equation*}
\mathcal L_{N_1}^{(1)}+\frac{1}{2} I_{N_1}=\mathcal L_{N_1+1}^{(1)}
\end{equation*}
and we can conclude.
\end{itemize}
\end{proof}

{We notice that, in Lemma \ref{l:2}, the zero-eigenvalue appears in the spectrum of $\mathcal L_{N}^{(1)}$ for the MC model and in the spectrum of $\mathcal L_{N_1}^{(1)}$ for the EB model. The matrix $\mathcal L_{N}^{(1)}$ is associated to the set $S$ of competent agents while the matrix $\mathcal L_{N_1}^{(1)}$ is associated to the set $U$ of incompetent agents. In the following result we will prove how at each time $t>0$ the EB model leads the system toward suboptimal decisions with respect to the optimal one given by the MC model at a fixed time.
 As we see in the following theorem, the zero-eigenvalue controls the asymptotic behavior of the system.
\begin{theorem}\label{th:1}
Let us consider a system of $N+N_1\in \mathbb N$, $N\ge1$, $N_1\ge1$ interacting agents, such that at a given time $S=\{1,\dots,N\}$ is the set of competent agents and $U=\{N+1,\dots,N+N_1\}$ is the set of not competent agents. In case of interaction function in Assumption \ref{ass:2} at each $t>0$ the collective decision in case of EB is not optimal.
\end{theorem}
\begin{proof}
Let us observe that for all $N>1$ we have $\dim (\ker \mathcal L_{N}^{(1)})=1$, in fact, the $(N-1)\times (N-1)$ top left minor in \eqref{eq:cor} has non zero determinant, while $\det (\mathcal L_{N}^{(1)})=0$ for each $N$. Thus, the rank of $\mathcal L_{N}^{(1)}$ is $N-1$.

In the following we denote with $m_a(\lambda)$ and $m_g(\lambda)$ the algebraic and the geometric multiplicity of eigenvalue $\lambda$. Now we can prove that the matrix $\mathcal L_{N+N_1}^{EB}$ is diagonalizable.

In Lemma \ref{l:2} we have shown that the eigenvalues of $\mathcal L_{N+N_1}^{EB}$ are $\lambda_1=\dots=\lambda_{N-1}= \frac{N}{2}+N_1$, $\lambda_ {N}=N_1$, $\lambda_{N+1}=\dots=\lambda_{N+N_1-1}= \frac{N_1}{2}$ and $\lambda_{N+N_1}=0$.

\begin{itemize}
\item
Case $m_g(0)=m_a(0)=1$. Notice that $\lambda=0$ belongs also to the spectrum of $\mathcal L_{N_1}^{(1)}$, $\dim( \ker \mathcal L_{N_1}^{(1)})=1$ and $(1,\dots,1)\in \ker \mathcal L_{N_1}^{(1)}$. Further we have that $\det\left(\mathcal L_{N}^{(1)}+N_1 Id_{N}\right)\neq 0$ and thus $\dim( \ker \mathcal L_{N+N_1}^{EB})=1$.

\item
Case $m_g\left(\frac{N_1}{2}\right)=m_a\left(\frac{N_1}{2}\right)=N_1-1$. We have:
\begin{equation*}
\newcommand*{\temp}{\multicolumn{1}{r|}{}}\mathcal L_{N+N_1}^{EB}-\frac{N_1}{2}\, Id_{N+N_1}=\left(
          \begin{array}{ccc}
            \mathcal L_{N}^{(1)}+\frac{N_1}{2} Id_{N}\quad & \temp &\quad - \mathcal J \\ \cline{1-3}
             \mathcal O & \temp &\quad  \mathcal L_{N_1}^{(1)}-\frac{N_1}{2} Id_{N_1} \\
          \end{array}
        \right)\, ,
\end{equation*}
where $\dim \ker\left(\mathcal L_{N_1}^{(1)}-\frac{N_1}{2} Id_{N_1}\right)=N_1-1$ and $\det\left(\mathcal L_{N}^{(1)}+\frac{N_1}{2} Id_{N}\right)\neq 0$. Therefore we have
\begin{equation*}
\dim \ker \left(\mathcal L_{N+N_1}^{EB}-\frac{N_1}{2}\, Id_{N+N_1}\right)=N_1-1.
\end{equation*}

\item Case $m_g\left(N_1\right)=m_a\left(N_1\right)=1$. Let us consider
\begin{equation*}
\newcommand*{\temp}{\multicolumn{1}{r|}{}}\mathcal L_{N+N_1}^{EB}-N_1\, Id_{N+N_1}=\left(
          \begin{array}{ccc}
            \mathcal L_{N}^{(1)}\quad & \temp &\quad -\mathcal J \\ \cline{1-3}
             \mathcal O & \temp &\quad  \mathcal L_{N_1}^{(1)}-N_1 Id_{N_1} \\
          \end{array}
        \right)\, ,
\end{equation*}
where the eigenvalues of $\mathcal L_{N_1}^{(1)}-N_1 Id_{N_1}$ are $-N_1$ and $-\frac{N_1}{2}$, thus $\det(\mathcal L_{N_1}^{(1)}$ $-N_1 Id_{N_1})\neq 0$. Moreover, $\dim (\ker \mathcal L_{N}^{(1)})=1$ and $(1,\dots,1)\in \ker \mathcal L_{N}^{(1)}$. Accordingly $\dim \ker \Bigl(\mathcal L_{N+N_1}^{EB}$ $-N_1\, Id_{N+N_1}\Bigr)=1$.

\item Case $m_g\left(\frac{N}{2}+N_1\right)=m_a\left(\frac{N}{2}+N_1\right)=N-1$. We now consider the matrix
\begin{equation*}\begin{split}
\mathcal L_{N+N_1}^{EB}&-\left(\frac{N}{2}+N_1\right)\, Id_{N+N_1}\\
&=\newcommand*{\temp}{\multicolumn{1}{r|}{}}\left(
          \begin{array}{ccc}
            \mathcal L_{N}^{(1)}- \frac{N}{2}\, Id_{N}\quad & \temp &\quad - \mathcal J\\ \cline{1-3}
             \mathcal O & \temp &\quad  \mathcal L_{N_1}^{(1)}-\left(\frac{N}{2}+N_1\right)\, Id_{N_1} \\
          \end{array}
        \right)\, ,
\end{split}\end{equation*}
where $\det\left(\mathcal L_{N_1}^{(1)}-\left(\frac{N}{2}+N_1\right)\, Id_{N_1}\right)\neq 0$ and $\dim \ker\left(\mathcal L_{N}^{(1)}-\frac{N}{2} Id_{N}\right)=N-1$. We have shown
\begin{equation*}
\dim \ker\left(\mathcal L_{N+N_1}^{EB}-\left(\frac{N}{2}+N_1\right)\, Id_{N+N_1}\right)=N-1\, .
\end{equation*}
\end{itemize}
The solution of the system of differential equations \eqref{eq:micro} is equivalent to a diagonal system with diagonal entries given by the eigenvalues in Lemma \ref{l:2}. Therefore, the collective decision in the EB case is not optimal.
\end{proof}
\begin{remark}
Observe that the case $N_1=0$ is encompassed in Corollary \ref{c:1} where, under Assumption \ref{ass:1}, all entries of the matrix $\mathcal R_N$ are $r_{ij}=\frac{1}{2}$ for each $i,j=1,\dots N$.
\end{remark}

\section{Numerics}\label{sec:num}
In this section we present several numerical results in order to show the main features of the    system \eqref{eq:micro} under the hypotheses of maximum competence and equality bias. We consider a set of $N=20$ agents, forming an interacting decision-making system. Therefore, we compare the emerging asymptotic collective decisions in the cMC and EB regimes for several $c\ge 1$.

For what it may concern the evolution of the competence variable we consider a background variable $z\in\RR^+$ with uniform distribution $C(z)\sim U([0,1])$. Further, the interaction function $\lambda(x_i,x_j)$ introduced in \eqref{eq:micro}, representing the possible learning processes of low skilled agents through the interaction with the more competent agents, is supposed to be
\be\label{eq:lambda_test}
\lambda(x_i,x_j) =
\begin{cases}
\bar\lambda & x_i<x_j, \\
0 & x_i\ge x_j.
\end{cases}
\ee
The numerics have been performed in the case $\bar\lambda=\lambda_B=10^{-2}$.
The ODE system \eqref{eq:micro} has been solved through the RK4 method by considering both for competence and decision the time step $\Delta t =10^{-2}$ and the final time $Tf = 10$. The interaction terms of the evolving decision have been chosen of the form  \eqref{eq:P} with $Q(w_i,w_j) = 1$ and $R(x_i,x_j)$ describing the cMC and EB cases for increasing values of the parameter $c>0$.\\

We consider a multi-agent system characterized at $t=0$ by decisions strongly clustered: the most competent agent with uniform distribution in $w\in[-1,-0.75]$, $x\in[0.75,1]$ and the less skilled agents with uniform distribution $w\in[0.75,1]$, $x\in[0,0.25]$. In all the tests the two populations of competent/incompetent agents are supposed to be of equal size. In Figure \ref{fig:EB_MC} and in Figure \ref{fig:traj} we compare the evolution of the system in the cMC case (blue line) and in the EB case (orange dashed line). The results are presented for $c=1,5,10$. We can observe how the collective decision of the system strongly diverges in the case of EB with respect to the optimal decision, given by cMC model with $c>>1$. A further evidence of the emerging suboptimality is given in Figure \ref{fig:distance} where we depict the asymptotic collective decision of the multi-agent system evolving in the cMC and EB cases and an increasing $c=1,\dots,10$.

\begin{figure}
\centering
\subfigure[c=1]{\includegraphics[scale=0.197]{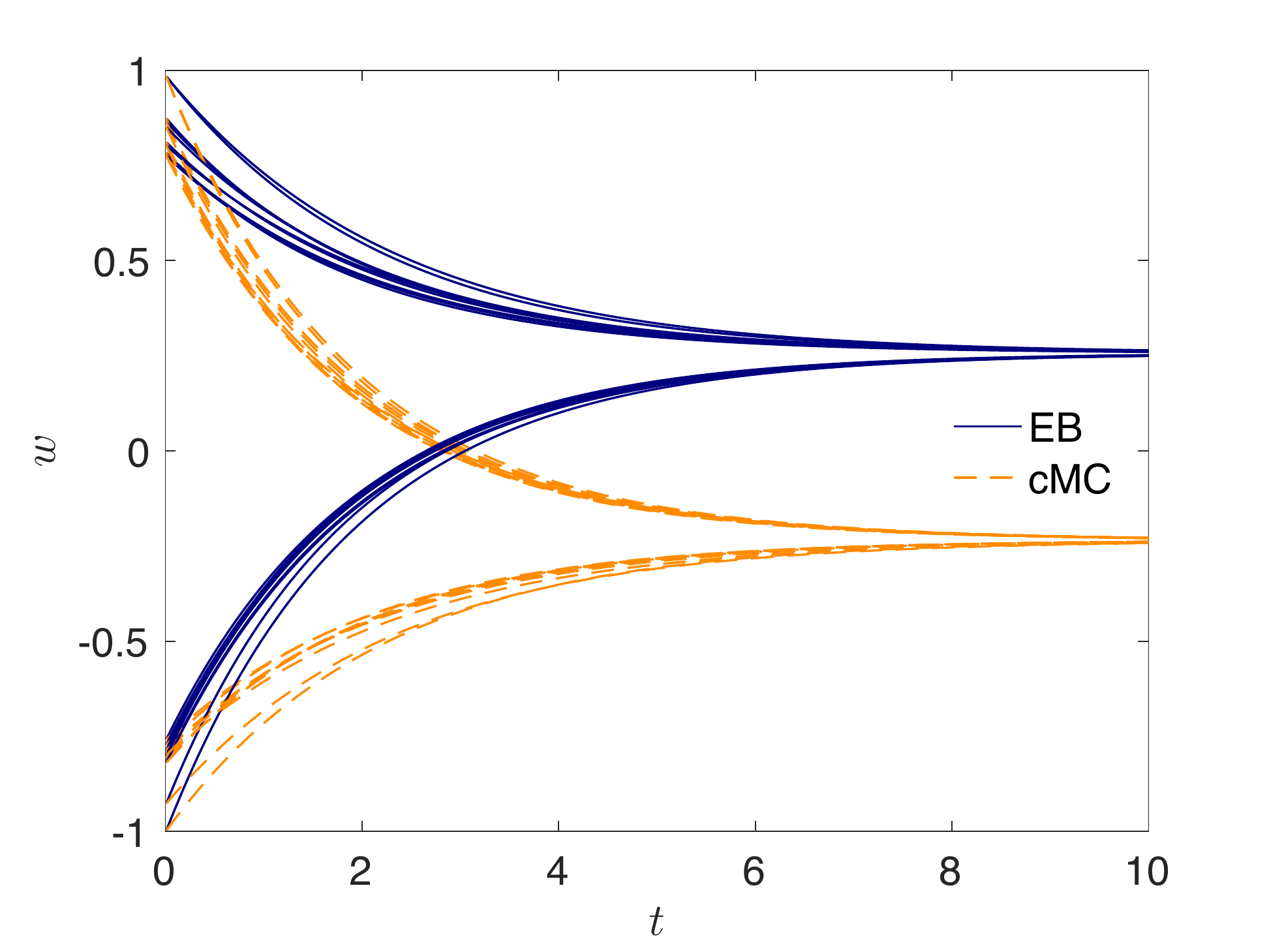}}
\subfigure[c=5]{\includegraphics[scale=0.197]{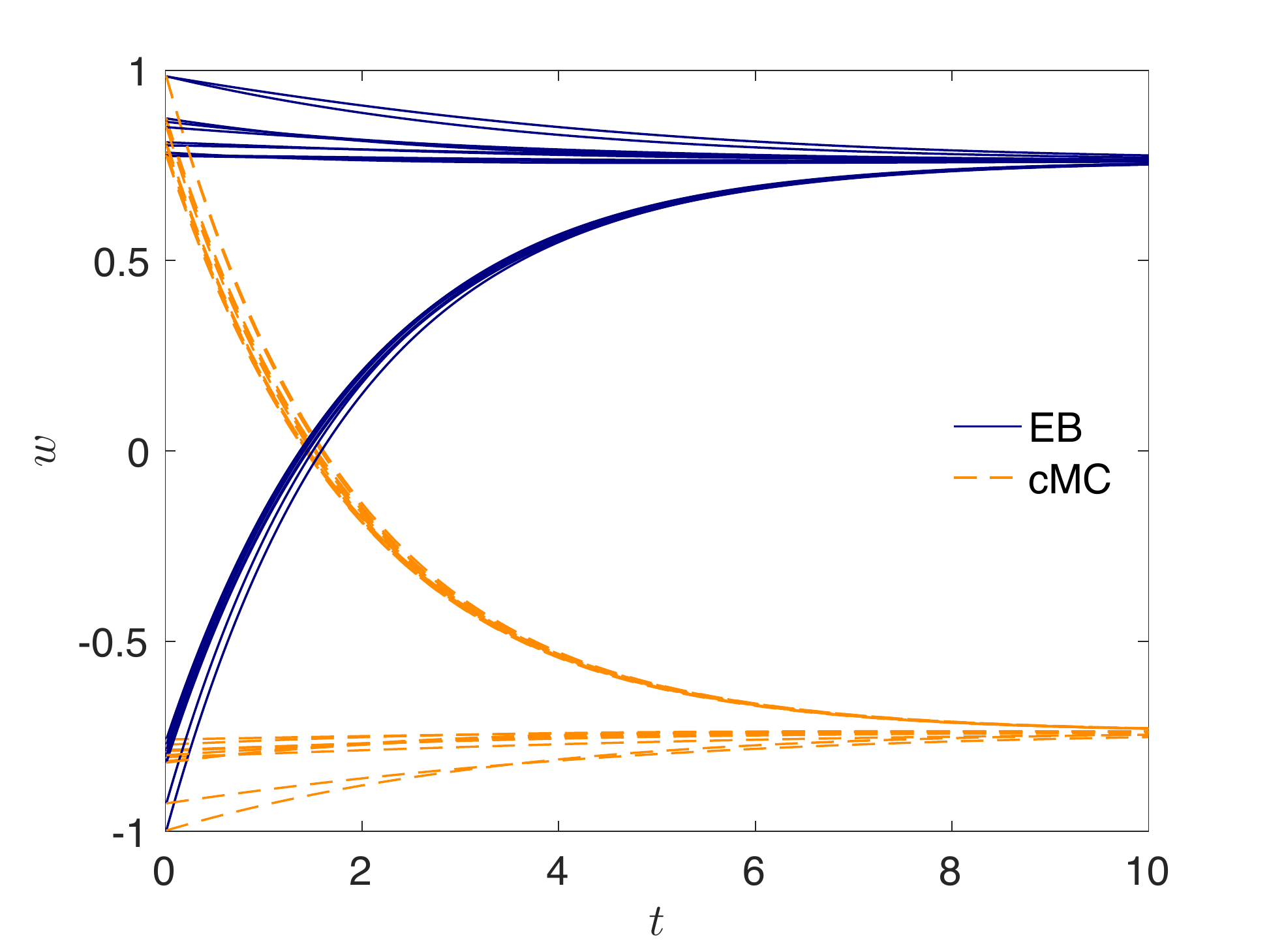}}
\subfigure[c=10]{\includegraphics[scale=0.197]{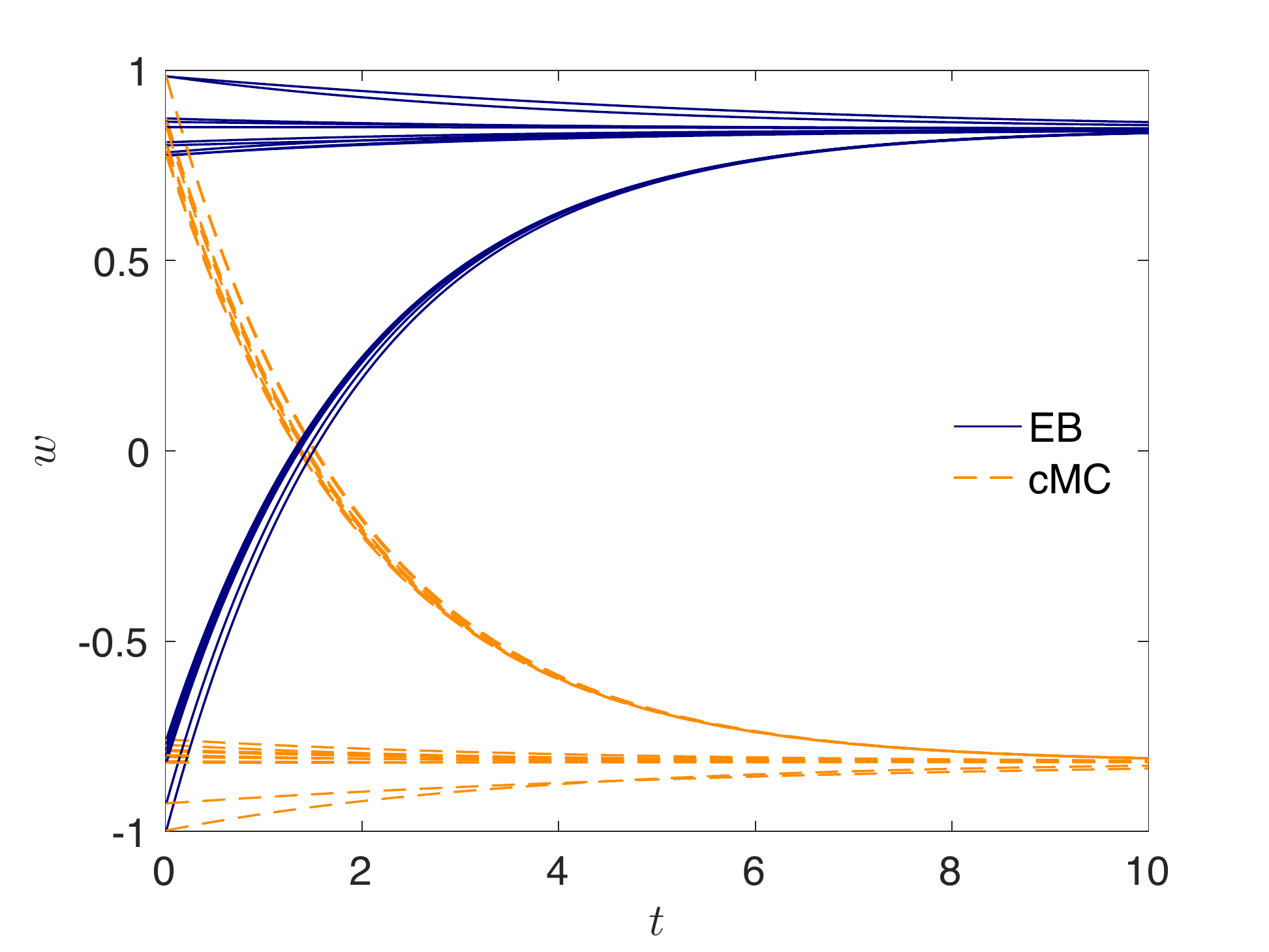}}
\caption{Evolution of the variable $w\in[-1,1]$, representing the decision of each agents, in the cases cMC and EB with increasing $c>0$, $N=20$.  The blue traits represent the evolution of the system in case of cMC interactions, whereas the orange dashed trait describe the evolution of the system under EB. We considered $\Delta t=10^{-2}$, $Tf=10$, $\bar{\lambda}=\lambda_B=10^{-2}$, solving the ODE system through RK4. }\label{fig:EB_MC}
\end{figure}

\begin{figure}
\centering
\subfigure[c=1]{\includegraphics[scale=0.197]{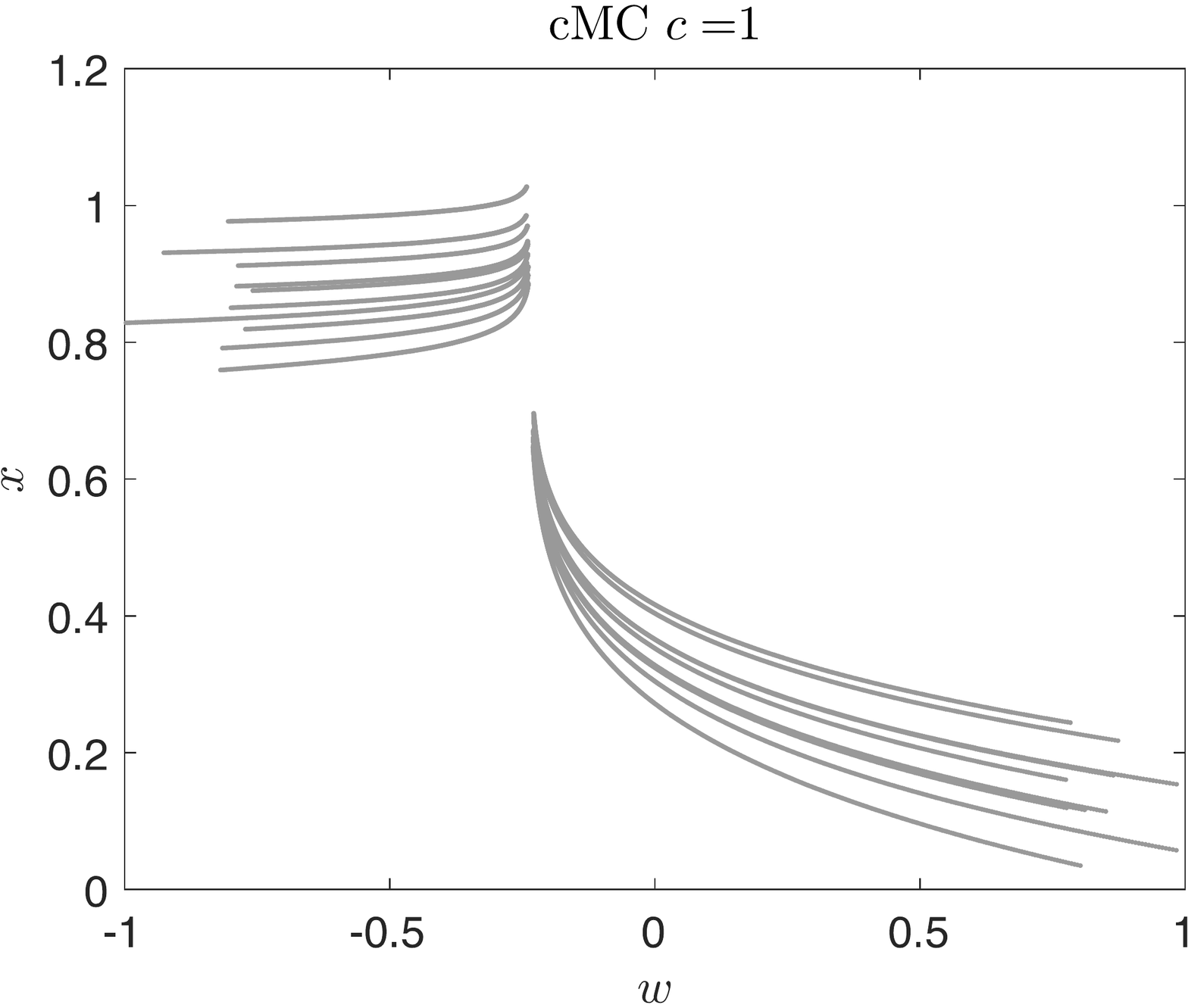}}
\subfigure[c=5]{\includegraphics[scale=0.197]{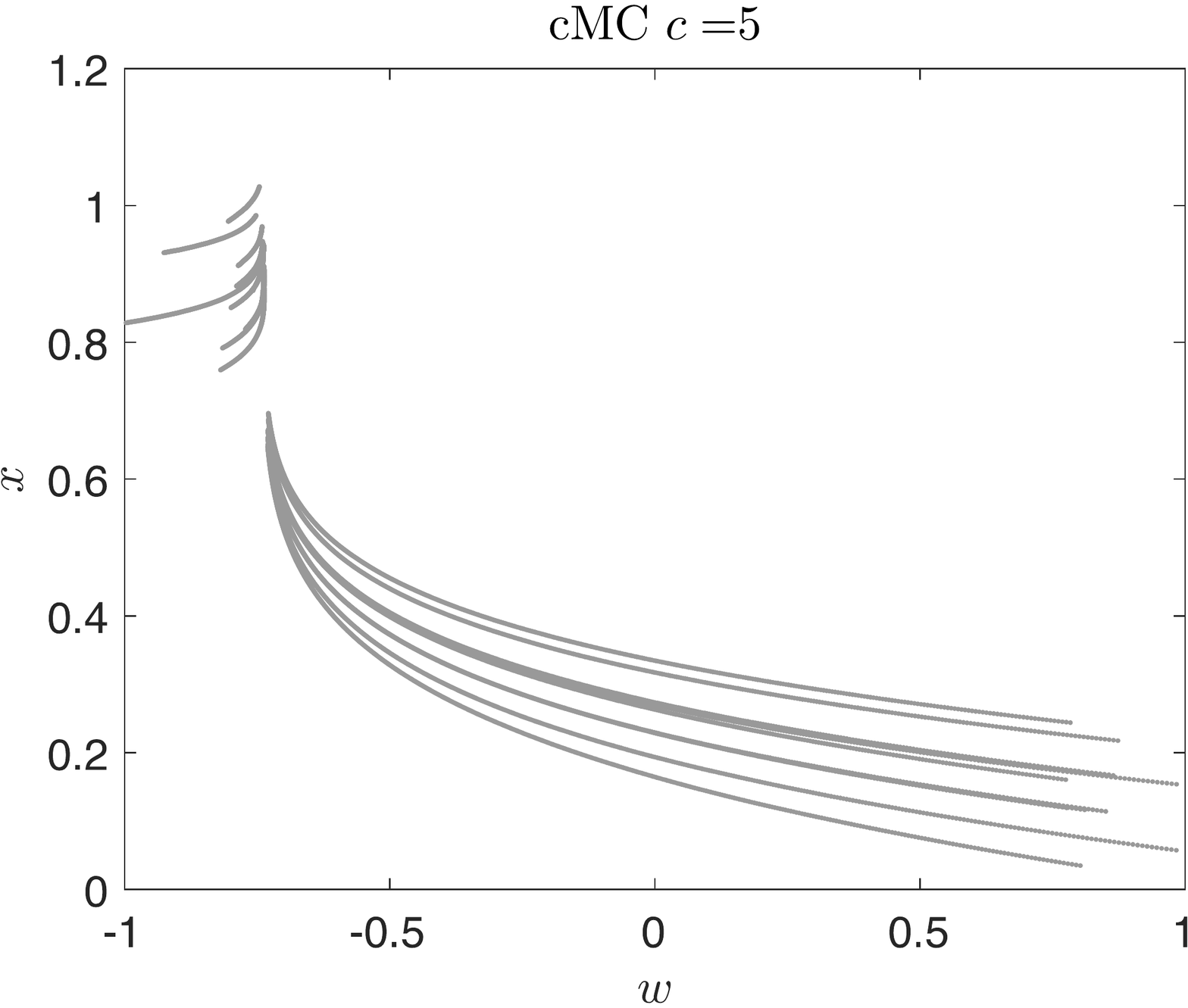}}
\subfigure[c=10]{\includegraphics[scale=0.197]{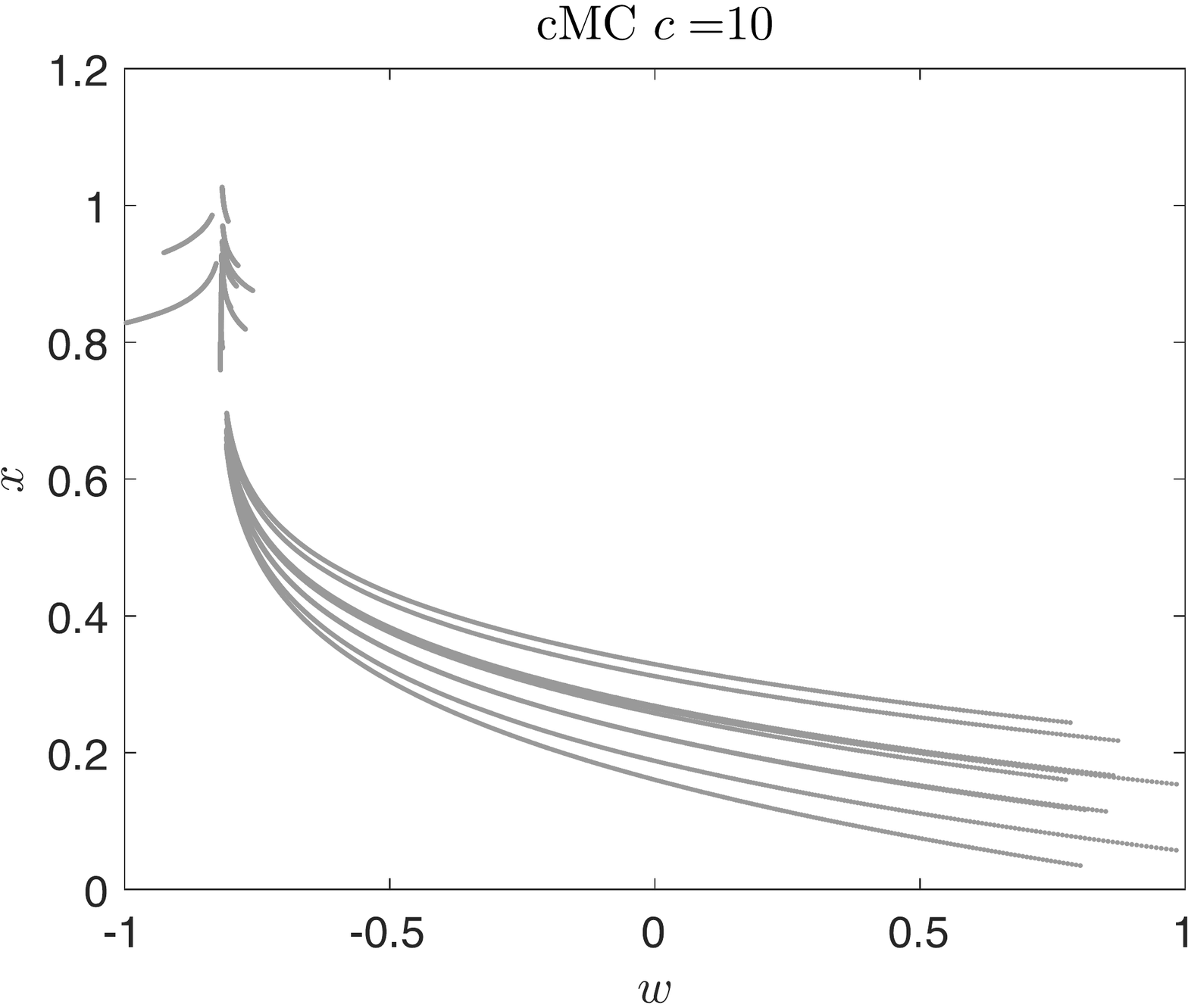}}\\
\subfigure[c=1]{\includegraphics[scale=0.197]{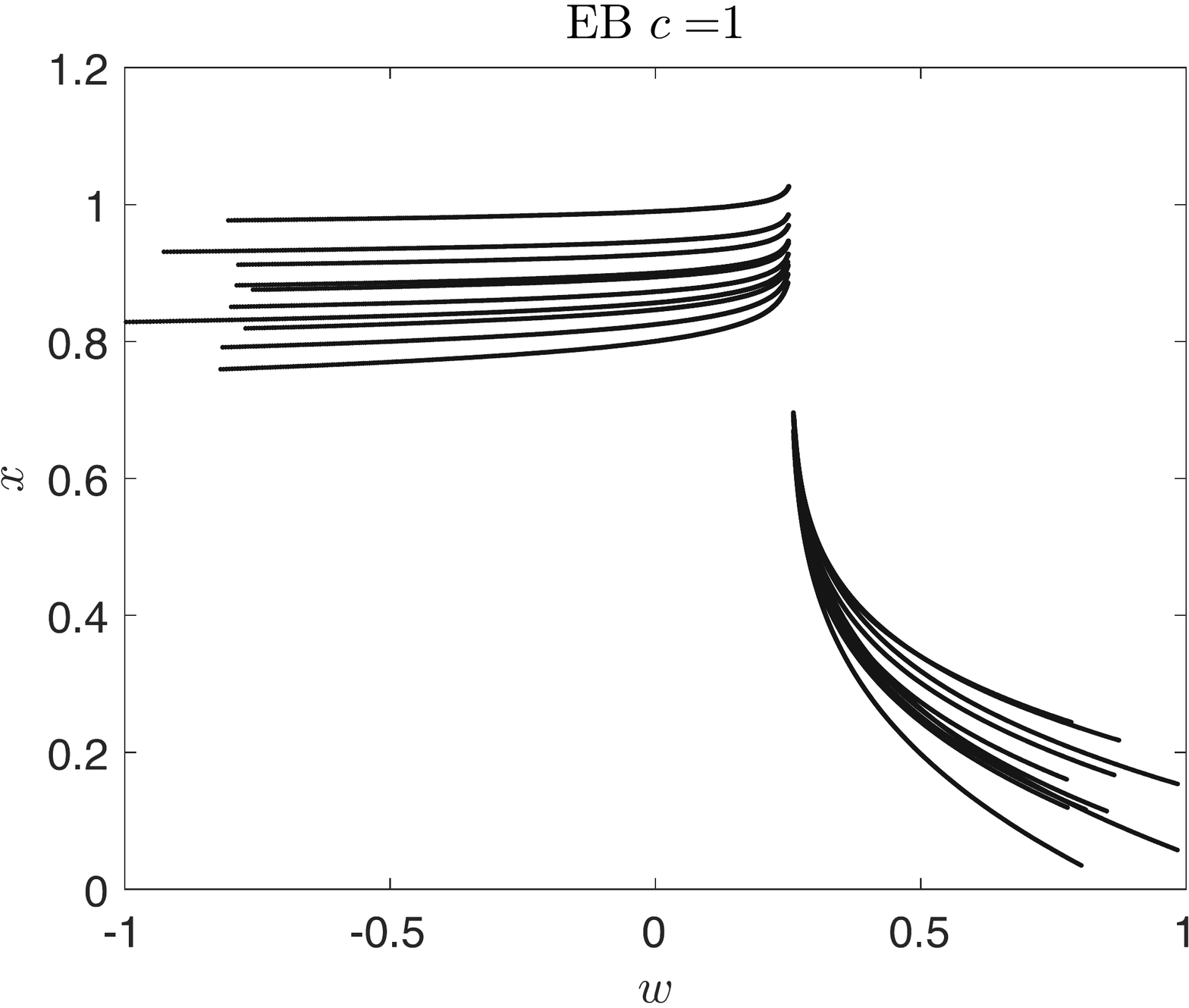}}
\subfigure[c=5]{\includegraphics[scale=0.197]{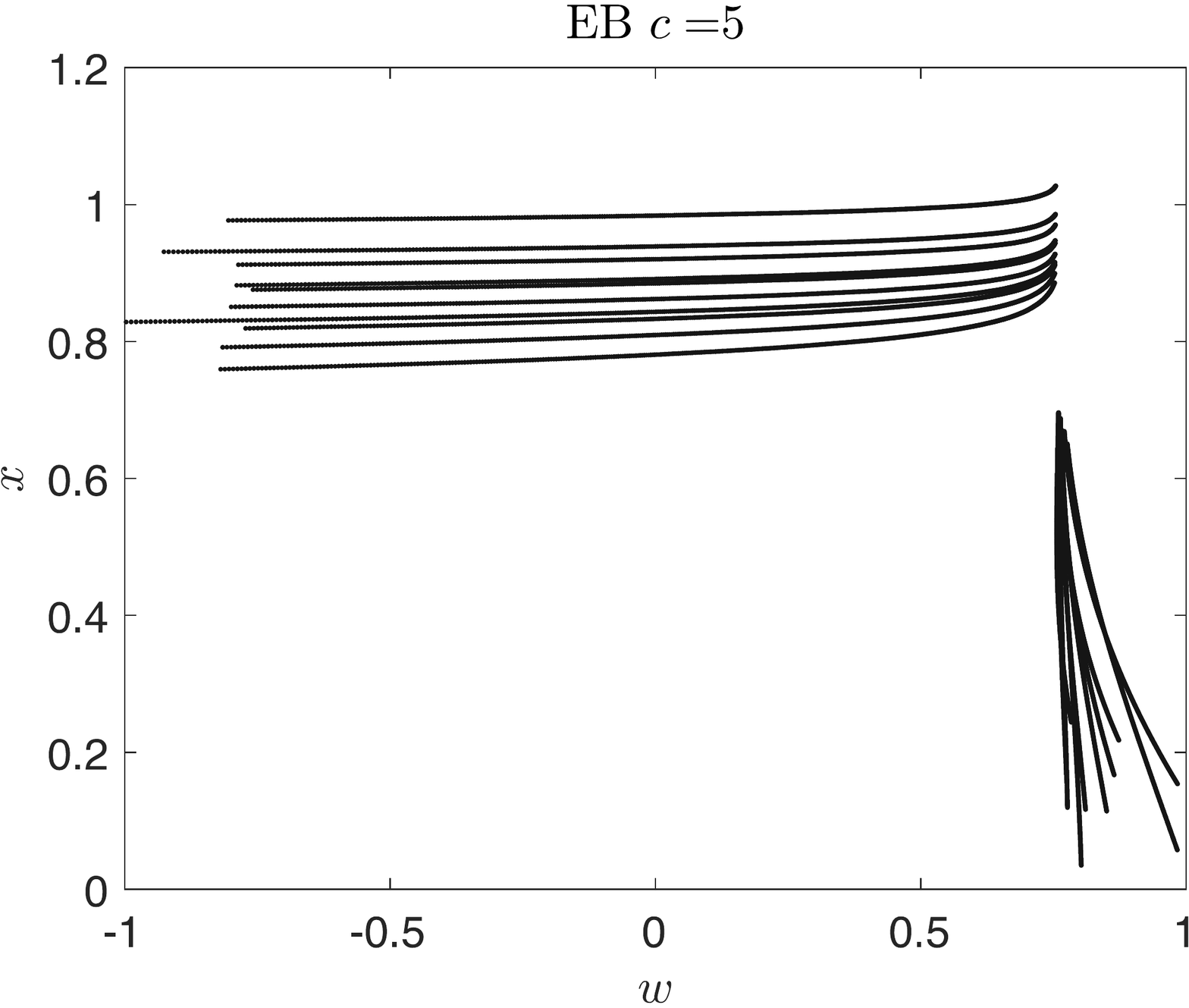}}
\subfigure[c=10]{\includegraphics[scale=0.197]{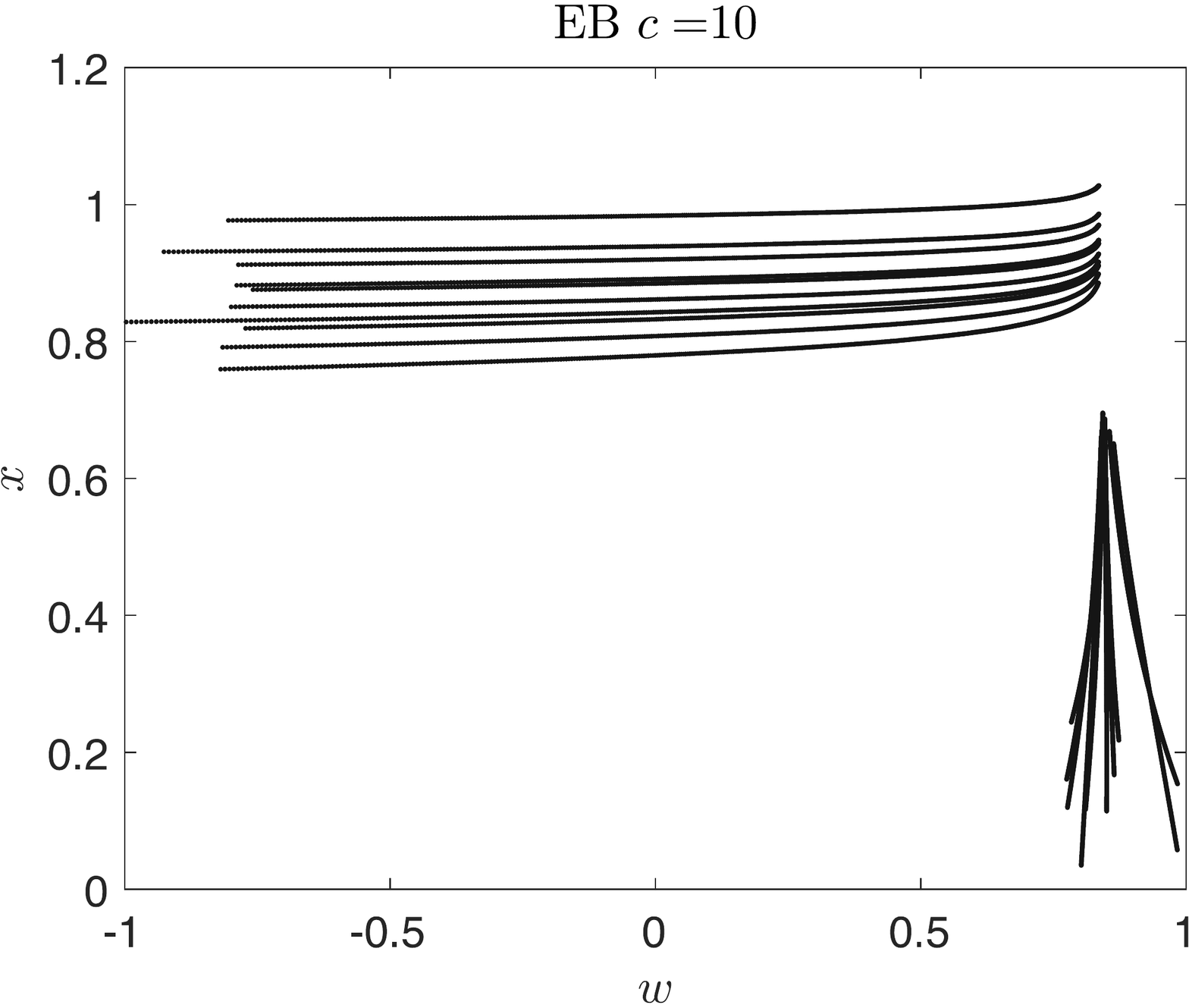}}
\caption{We illustrate in these figures the competence-decision dynamics of $N=20$ particles as we modeled in \eqref{eq:micro} considered in the whole time interval $[0,10]$, $\Delta t=10^{-2}$. The initial disposition of the competence and decision variables are uniform in $(x,w)\in [0.75,1]\times [-1,-0.75]$ and $(x,w)\in[0,0.25]\times [0.75,1]$. In the figure $(a)-(b)-(c)$ we represent the complete trajectories in the case of cMC with increasing $c>0$. In $(d)-(e)-(f)$ we present the dynamics in the EB case for several values of $c>0$. The evolution of the competence follows $\lambda(x_i,x_j)$ in \eqref{eq:lambda_test} with $\bar{\lambda}=\lambda_B=10^{-2}$. }\label{fig:traj}
\end{figure}

\begin{figure}
\centering
\includegraphics[scale=0.4]{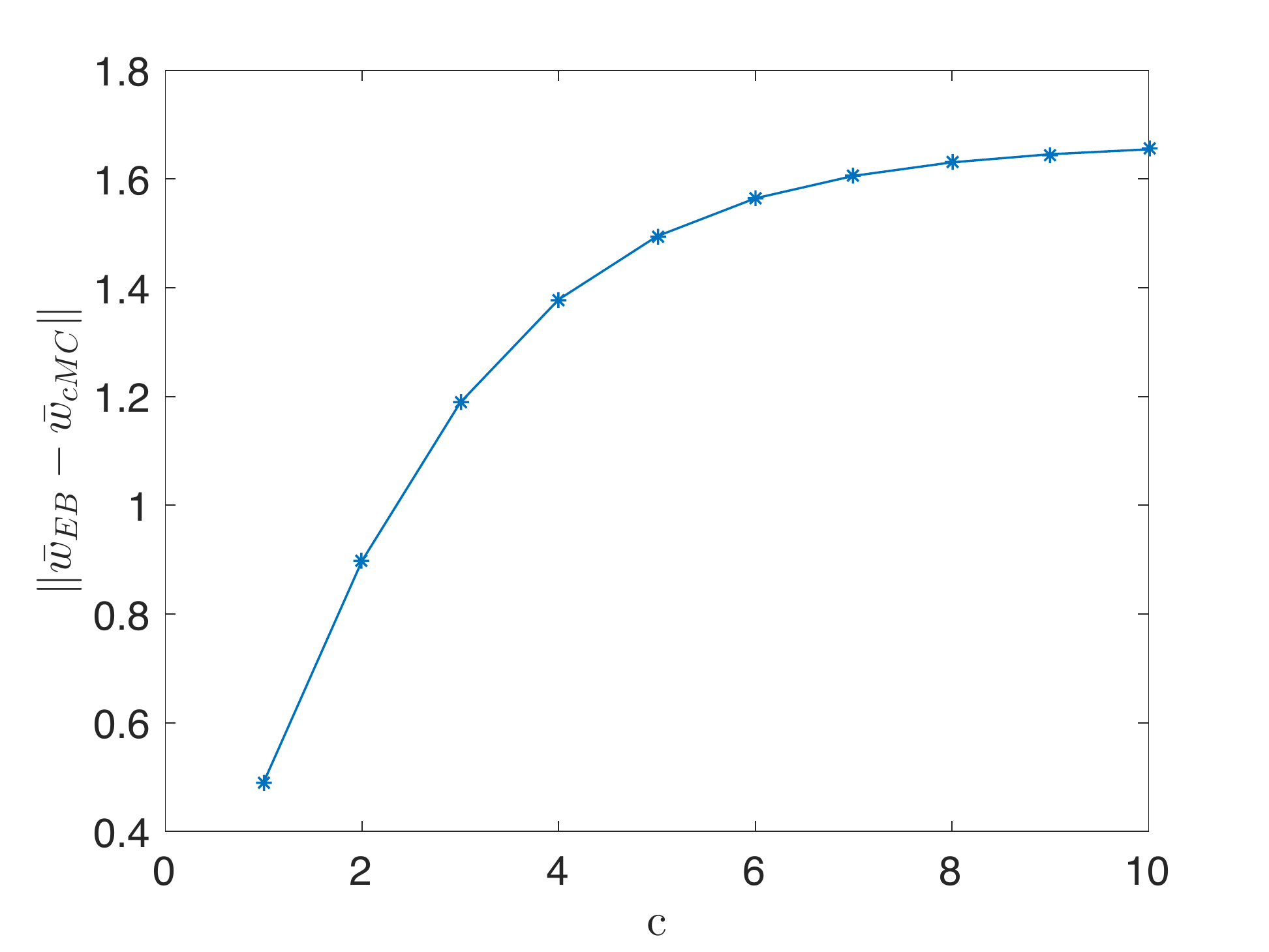}
\caption{Distance in $\ell^1$ norm of the asymptotic collective decisions in the cMC and EB case, i.e. $\bar w_{EB}$ and $\bar w_{cMC}$, with respect to $c=1,\dots,10$. }\label{fig:distance}
\end{figure}

\section*{Conclusions}
We introduced microscopic models of interacting multi-agent systems for the description of decision-making processes inspired by the experimental results in \cite{BOLRRF,MPNAS}. The introduction of the bias is here analytically studied at the level of the Laplacian matrix. We explicitly calculated the structure of the eigenvalues of the Laplacian matrix under under simplified assumptions both in the maximum competence and in the equality bias case. The suboptimality of the collective decision under equality bias with respect to the maximum competence case is then established for each time step. Numerical results show that the equality bias impairs the emergence of the decision of the most competent agents of the system.

\section*{Acknowledgements}
PV would like to acknowledge Dr. Luigi Teodonio for his stimulating discussions about the topic of the paper. MZ acknowledges the ''Compagnia di San Paolo''.


\begin{thebibliography}{9}
\bibitem{AP}
G. Albi, L. Pareschi. Modeling of self-organized systems interacting with a few individuals: from microscopic to macroscopic dynamics. \emph{Applied Mathematics Letters} 26(4): 397--401, 2013.

\bibitem{AHP} G.~Albi, M.~Herty, L.~Pareschi. Kinetic description of optimal control problems and applications to opinion consensus. \emph{Communications in Mathematical Sciences}, 13(6): 1407--1429, 2015.

\bibitem{APTZ}
G. Albi, L. Pareschi, G. Toscani, M. Zanella. Recent advances in opinion modeling: control and social influence. In \emph{Active Particles Volume 1, Theory, Methods, and Applications}, N. Bellomo, P. Degond, and E. Tadmor Eds., Birkh{\"a}user-Springer, 2017.

\bibitem{APZa} G.~Albi, L.~Pareschi, M.~Zanella. Boltzmann-type control of opinion consensus through leaders. \emph{Philosophical Transactions of the Royal Society of London A: Mathematical, Physical and Engineering Sciences}, 372(2028): 20140138, 2014.

\bibitem{APZb}
G. Albi, L. Pareschi, M. Zanella. Uncertainty quantification in control problems for flocking models. \emph{Mathematical Problems in Engineering}, Vol. 2015, 14 pp., 2015.

\bibitem{APZc}
 G.~Albi, L.~Pareschi, M.~Zanella. Opinion dynamics over complex networks: kinetic modeling and numerical methods. \emph{Kinetic and Related Models}, to appear.



\bibitem{BOLRRF}
B. Bahrami, K. Olsen, P. E. Latham, A. Roepstorff, G. Rees, C. D. Frith. Optimally interacting minds. \emph{Science}, 329(5995): 1081--1085, 2010.
%

\bibitem{BT}
C. Brugna, G. Toscani. Kinetic models of opinion formation in the presence of personal conviction. \emph{Physical Review E}, 92(5): 052818, 2015.

\bibitem{CCR}
J. A. Ca\~{n}izo, J. A. Carrillo, J. Rosado. A well-posedness theory in measures for some kinetic models of collective behavior. \emph{Mathematical Models and Methods in Applied Sciences} 21(3): 515--539, 2011.

\bibitem{CFRT}
J. A. Carrillo, M. Fornasier, J. Rosado, G. Toscani. Asymptotic flocking dynamics for the kinetic Cucker-Smale model. \emph{SIAM Journal on Mathematical Analysis}, 42(1): 218--236, 2010.

\bibitem{CFTV}
J. A. Carrillo, M. Fornasier, G. Toscani, F. Vecil. Particle, kinetic and hydrodynamic models of swarming. In \emph{Mathematical Modeling of Collective Behavior in Socio-Economic and Life Sciences}, G. Naldi, L. Pareschi, and G. Toscani Eds., Birkh{\'a}user Boston, pp. 297--336, 2010.

\bibitem{CFL}
C. Castellano, S. Fortunato, V. Loreto. Statistical physics of social dynamics. \emph{Reviews of Modern Physics} 81(2): 591, 2009.

\bibitem{CC}
A. Chakraborti, B. K. Chakrabarti, Statistical mechanics of money: how saving propensity affects its distribution, \emph{The European Physical Journal B-Condensed Matter and Complex Systems} 17(1): 167--170, 2000.

\bibitem{CPT}
E. Cristiani, B. Piccoli, A. Tosin. \emph{Multiscale Modeling of Pedestrian Dynamics}, MS\&A: Modeling, Simulation and Applications, vol.12, Springer International Publishing, 2014.

\bibitem{CS}
F. Cucker, S. Smale. Emergent behavior in flocks. \emph{IEEE Transactions on Automatic Control}, 52(5): 852--862, 2007.

\bibitem{DOCBC}
M. R. D'Orsogna, Y. L. Chuang, A. L. Bertozzi, L. S. Chayes. Self-propelled particles with soft-core interactions: patterns, stability and collapse. \emph{Physical Review Letters} 96(10):104--302, 2006.

\bibitem{DMPW}
B. D\"{u}ring, P. Markowich, J. F. Pietschmann, M.-T. Wolfram. Boltzmann and Fokker-Planck equations modelling opinion formation in the presence of strong leaders. \emph{Proceedings of the Royal Society of London A: Mathematical, Physical and Engineering Sciences}  465(2112): 3687--3708, 2009.

\bibitem{G97}
S. Galam. Rational group decision making: a random field Ising model at T=0. \emph{Physica A. Statistical Mechanics and its Applications}, 238(1--4): 66--80, 1997.

\bibitem{GZ}
S. Galam, J.-D. Zucker. From individual choice to group decision-making. \emph{Physica A. Statistical Mechanics and its Applications}, 287: 644--659, 2000.

\bibitem{G}
F. Galton. One vote, one value. \emph{Nature} 75: 414--414, 1907.

\bibitem{HF}
N. Harvey, I. Fischer. Taking advice: Accepting help, improving judgment, and sharing responsibility. \emph{Organizational Behavior and Human Decision Processes} 70(2): 117--133, 1997.

\bibitem{HZ}
M. Herty, M. Zanella. Performance bounds for the mean-field limit of constrained dynamics. \emph{Discrete and Continuous Dynamical Systems A}, 37(4): 2023--2043, 2017.

\bibitem{KD}
J. Kruger, D. Dunning. Unskilled and unaware of it: how difficulties in recognizing one's incompetence lead to inflated self-assessments. \emph{Journal of Personality and SOcial Psychology}, 77(6): 1121--1134, 1999.


\bibitem{MPNAS}
A. Mahmoodi, D. Bang, K. Olsen, Y. A. Zhao, Z. Shi, K. Broberg, S. Safavi, S. Han, M. N. Ahmadabadi, C. D. Frith, A. Roepstorff, G. Rees, B. Bahrami. Equality bias impairs collective decision-making across cultures. \emph{Proceedings of the National Academy of Sciences} 112(12): 3835--3840, 2015.


\bibitem{MT1}
S. Motsch, E. Tadmor. Heterophilious dynamics enhances consensus. \emph{SIAM Review} 56(4): 577--621, 2014.

\bibitem{MT2}
S. Motsch, E. Tadmor. A new model for self-organized dynamics and its flocking behavior. \emph{Journal of Statistical Physics} 144(5): 923--947, 2011.


\bibitem{PT1}
L. Pareschi, G. Toscani. Wealth distribution and collective knowledge: a Boltzmann approach. \emph{Philosophical Transactions of the Royal Society A: Mathematical, Physical and Engineering Sciences} 372(2028): 20130396, 2014.

\bibitem{PT2}
L. Pareschi, G. Toscani. \emph{Interacting Multiagent Systems. Kinetic Equations and Monte Carlo Methods}. Oxford University Press, 2013.

\bibitem{PVZ}
L. Pareschi, P. Vellucci, M. Zanella. Kinetic models of collective decision-making in the presence of equality bias. \emph{Physica A. Statistical Mechanics and its Applications}, 467: 201--217, 2017.


\bibitem{Shen}
J. Shen. Cucker--Smale flocking under hierarchical leadership. \emph{SIAM Journal on Applied Mathematics} 68(3): 694--719, 2007.

\bibitem{T}
G. Toscani. Kinetic models of opinion formation. \emph{Communications in Mathematical Sciences} 4(3): 481--496, 2006.
\end{thebibliography}
\end{document}